\documentclass[a4paper,11pt]{article}
\pdfoutput=1

\usepackage[margin=0.8in]{geometry}

\usepackage[utf8]{inputenc}
\usepackage[british]{babel}
\usepackage{csquotes} 
\usepackage[T1]{fontenc}

\usepackage{pifont}

\usepackage{amsmath,amssymb,amsthm}
\usepackage{mathtools}

\usepackage[
    labelfont=bf %
]{caption}
\usepackage{float}
\usepackage{subcaption}
\usepackage[most]{tcolorbox}

\usepackage{tikz}
\usetikzlibrary{decorations.markings,arrows,patterns,patterns.meta,shapes}
\usetikzlibrary{positioning}

\usepackage{pgfplots}
\pgfplotsset{compat=1.18}
\pgfdeclareplotmark{donut}{
    \pgfseteorule
    \pgfpathellipse{\pgfpoint{0cm}{0cm}}
             {\pgfpoint{3pt}{0cm}}
             {\pgfpoint{0cm}{1pt}}
     \pgfpathellipse{\pgfpoint{0cm}{0cm}}
             {\pgfpoint{4pt}{0cm}}
             {\pgfpoint{0cm}{2pt}}
    \pgfusepath{fill}
}
\pgfdeclareplotmark{donut*}{
     \pgfpathellipse{\pgfpoint{0cm}{0cm}}
             {\pgfpoint{4pt}{0cm}}
             {\pgfpoint{0cm}{2pt}}
    \pgfusepath{fill}
}

\usepackage[outline]{contour}
\contourlength{0.2em}

\usepackage[compat=0.6]{yquant}
\useyquantlanguage{groups}

\usepackage{url}

\usepackage{mathrsfs}
\usepackage{dsfont}
\usepackage{setspace}
\usepackage[export]{adjustbox}
\usepackage{makecell}
\usepackage{tqft}
\usepackage[inline]{enumitem}

\usepackage{quantikz}
\usepackage{graphicx}

\usepackage{xcolor}
\usepackage{bm}
\definecolor{darkblue}{RGB}{0,0,128}
\definecolor{darkgreen}{RGB}{0,150,0}

\definecolor{dark-red}{rgb}{0.4,0.15,0.15}
\definecolor{dark-blue}{rgb}{0.15,0.15,0.4}
\definecolor{dark-green}{rgb}{0.15,0.4,0.15}
\definecolor{medium-blue}{rgb}{0,0,0.5}

\definecolor{drawing-blue}{HTML}{c1ecff}
\definecolor{drawing-yellow}{HTML}{f6dab5}
\definecolor{drawing-red}{HTML}{e5bdd2}

\usepackage[%
	backend=biber,
	style=alphabetic,
	url=false,
	isbn=true,
	maxbibnames=10,
	backref, 
]{biblatex}
\usepackage[pdfusetitle]{hyperref}
\hypersetup{
    breaklinks,
    colorlinks,
	linkcolor={dark-blue},
	citecolor={dark-green},
	urlcolor={medium-blue},
    filecolor=red,
}

\usepackage[nameinlink,capitalise]{cleveref}
\crefname{figure}{Figure}{Figures}

\crefname{conjecture}{Conjecture}{Conjectures}
\Crefname{conjecture}{Conjecture}{Conjectures}

\newtheorem{theorem}{Theorem}
\newtheorem*{theorem*}{Theorem}

\newtheorem{claim}[theorem]{Claim}

\usepackage{thmtools}
\usepackage{thm-restate}

\usepackage{braket}
\usepackage{dsfont}
\usepackage[normalem]{ulem}
\usepackage{algorithm}
\usepackage{algpseudocode}

\usepackage{stmaryrd}

\def\Tr{\operatorname{Tr}}

\DeclareMathOperator{\tr}{tr}
\DeclareMathOperator{\CNOT}{CNOT}
\DeclareMathOperator{\CCZ}{CCZ}
\newcommand{\CSS}{\ensuremath{\mathsf{CSS}}}
\DeclareMathOperator{\GRS}{\mathsf{GRS}}
\DeclareMathOperator{\balpha}{\boldsymbol{\alpha}}
\newcommand{\bv}{\ensuremath{\boldsymbol{v}}}

\newcommand{\sansserif}[1]{%
  \ifmmode
  \mathsf{#1}%
  \else
   \textsf{#1}%
  \fi
}

\usepackage[mathscr]{euscript}

\usepackage{mathtools}

\DeclarePairedDelimiterXPP\bigo[1]{O}{(}{)}{}{#1}
\DeclarePairedDelimiterXPP\littleo[1]{o}{(}{)}{}{#1}
\DeclarePairedDelimiterXPP\bigomega[1]{$\Omega$}{(}{)}{}{#1}
\DeclarePairedDelimiterXPP\bigtheta[1]{$\Theta$}{(}{)}{}{#1}

\graphicspath{{../}{./}{./tikz}{figs/}}

\renewcommand{\sec}[1]{\hyperref[sec:#1]{Section~\ref*{sec:#1}}}
\newcommand{\app}[1]{\hyperref[app:#1]{Appendix~\ref*{app:#1}}}
\newcommand{\ssec}[1]{\hyperref[ssec:#1]{Subsection~\ref*{ssec:#1}}}
\newcommand{\fig}[1]{\hyperref[fig:#1]{Figure~\ref*{fig:#1}}}
\newcommand{\tab}[1]{\hyperref[tab:#1]{Table~\ref*{tab:#1}}}
\newcommand{\lemm}[1]{\hyperref[lemm:#1]{Lemma~\ref*{lemm:#1}}}
\newcommand{\propos}[1]{\hyperref[propos:#1]{Proposition~\ref*{propos:#1}}}
\newcommand{\thm}[1]{\hyperref[thm:#1]{Theorem~\ref*{thm:#1}}}
\newcommand{\alg}[1]{\hyperref[alg:#1]{Algorithm~\ref*{alg:#1}}}
\newcommand{\defn}[1]{\hyperref[defn:#1]{Definition~\ref*{defn:#1}}}

\newcommand{\be}{\begin{equation}}
\newcommand{\ee}{\end{equation}}
\newcommand{\bea}{\begin{eqnarray}}
\newcommand{\eea}{\end{eqnarray}}

\newcommand{\bem}{\begin{multline}}
\newcommand{\eem}{\end{multline}}

\usepackage{array}
\usepackage{booktabs}
\usepackage{makecell}
\usepackage{diagbox}
\usepackage{multicol}
\usepackage{multirow}

\usepackage{siunitx}

\newcolumntype{?}{!{\vrule width 1pt}}

\usepackage{authblk}

\begin{document}

\title{A Review of Galois Qudits}
\author[1]{Adam Wills\thanks{\texttt{a\_wills@mit.edu}}}
\affil[1]{Center for Theoretical Physics --- a Leinweber Institute\\

Massachusetts Institute of Technology, Cambridge, MA}

\date{\today}

\maketitle
\vspace*{-0.8cm}
\begin{abstract}
Galois qudits are $q$-dimensional quantum systems whose choice of Pauli group encodes the arithmetic of some finite field $\mathbb{F}_q$. They differ from the more familiar modular qudit, which are the same quantum system but whose choice of Pauli group are the clock and shift operators, which encode the arithmetic of integer addition and multiplication modulo $q$. Galois qudits are a useful mathematical construct that allow us to leverage the mathematical tools that are native to the larger qudit while only physically building smaller qudits. In particular, a Galois qudit of dimension $q=2^s$ is exactly the same thing as a collection of $s$ qubits, not only in its Hilbert space, but also in its Pauli group, and Clifford hierarchy. This formalism has found a lot of utility recently in constructing quantum error-correcting codes over qubits with useful properties.

In this review, we build on existing literature to collect and formalise facts and proofs about Galois qudits over binary extensions fields. We define them and their Clifford hierarchies, describe what it means to measure their Pauli operators, describe their stabiliser tableaux, formally define qudit-to-qubit mappings, and finally describe quantum Reed-Solomon codes.

\end{abstract}

\setcounter{tocdepth}{2}
\tableofcontents

\newpage

\section{Background and Outline}

A $q$-dimensional Galois qudit is a quantum system with Hilbert space $\mathbb{C}^q$, with a particular choice of Pauli group (encoding the arithmetic of the finite field $\mathbb{F}_q$). Technically, therefore, the name Galois qudit is a bit of a misnomer, because when we say ``Galois qudit'', we are talking not only about the qudit as a quantum system, but about our choice of Pauli group for it.

Non-binary quantum codes were first studied by Rains~\cite{rains1997nonbinary,rains2002nonbinary}, leading to a fruitful line of work studying quantum codes over $\mathbb{F}_q$, that is, quantum codes for $q$-dimensional Galois qudits. There was a fairly large amount of attention on this earlier on~\cite{ashikhmin2001asymptotically,matsumoto2001algebraic,ashikhmin2002nonbinary,grassl2003efficient,niehage2005quantum,ketkar2006nonbinary,aly2006primitive,niehage2007nonbinary,aly2009constructions}, but these ideas appear to have fallen out of fashion in the $2010$s,\footnote{A line of work has studied entanglement-assisted quantum code over finite fields,~\cite{riguang2011non,guenda2018constructions,galindo2019entanglement,fan2022entanglement,fan2025entanglement,sok2022linear,nadkarni2021entanglement}, which is an interesting use of the Galois qudit formalism. In this review, we are interested in regular (non-entanglement-assisted) quantum codes.} possibly due to the difficulties experienced in engineering any qudit larger than a qubit in hardware. They have recently seen a revival, however~\cite{wills2025constant,golowich2025asymptotically,nguyen2025good,nguyen2025quantum,golowich2025quantum,he2025quantum,he2025asymptotically,golowich2025near}, sparked primarily by the ability to convert codes over $\mathbb{F}_{q=2^s}$ with desirable properties to codes over qubits with similarly desirable properties. Simultaneously, these codes have come to be known colloquially as ``Galois qudit codes''~\cite{eczoo_galois_into_galois}, and the individual qudits over $\mathbb{F}_q$ as ``Galois qudits'', although this terminology has not appeared in the literature until now.


We will now aim to collect known facts in the literature about Galois qudits and quantum codes over finite fields. Our inspiration particularly comes from Chapter 8 of~\cite{gottesman2024surviving}, although we add formalism and proofs where we find it to be appropriate. We will assume knowledge of qubit CSS codes, see for example~\cite{nielsen2010quantum}.

We will begin in Section~\ref{sec:finite_fields} by exposing finite field and Galois qudits. There will be no reference to qubits at this stage; the qudits will be treated as an object in their own right. Then in Section~\ref{sec:qudit_to_qubit} we will show how Galois qudits can be treated as sets of multiple qubits. Finally, in Section~\ref{sec:qRS}, we will discuss quantum Reed-Solomon codes.

\section{Finite Fields and Galois Qudits}\label{sec:finite_fields}

\subsection{Finite Fields}

What is a field? Formal definitions are available, but simply put a field is a set of symbols on which we can do the basic arithmetic operations of addition, subtraction, multiplication, and division in all the ways that we are used to; multiplication distributes over addition, we can divide by non-zero elements, and so on. We are quite used to infinite fields like the real numbers and the complex numbers. One can ask, however, whether there exist fields with only finitely many elements.

One can see quite quickly that the answer is yes, because you can write down the numbers $\{0, 1\}$ and form a field just by performing all your usual arithmetic modulo $2$. We'll denote this field $\mathbb{F}_2$. In fact, you can write down the numbers $\{0, 1, \ldots, p-1\}$, for any prime $p$, and by performing all arithmetic modulo $p$, you can form a field, which we will denote $\mathbb{F}_p$. Let's try and do the same for some non-prime number. Consider the set $\mathbb{Z}_4 = \{0,1,2,3\}$, and consider performing addition and multiplication modulo $4$. One can see that this is not a field because, for example, in a field every non-zero element should have a multiplicative inverse (we should be able to divide by every non-zero element). That is to say that for every $x \neq 0$ there exists some element $x^{-1}$ such that $x \cdot x^{-1} = 1$. One can check that this is not the case for $2 \in \mathbb{Z}_4$.

It is natural to wonder if we can classify all the finite fields. Galois was the first to do so, and so finite fields are often called Galois fields. It turns out that there is a finite field with $n$ elements if and only if $n$ is a prime power, meaning that $n = p^s$, where $p$ is a prime number and $s \geq 1$ is an integer. Moreover, it turns out that if $n$ is a prime power, then there is exactly one field with $n$ elements \textit{up to isomorphism}.\footnote{This phrase ``up to isomorphism'' means exactly what you think it means. $8 = 2^3$ is a prime power, so there exists a finite field of order $8$. You could write down this finite field in infinitely many different ways, calling the elements infinitely many different things, but ultimately they're all exactly the same up to some re-labelling.} In particular, this means that there is exactly one field of order $4$, up to isomorphism, but it is not simply $\mathbb{Z}_4$.

Canonically, we write prime powers with the letter $q = p^s$. From now on, we will be exclusively interested in the case $p = 2$, and $q = 2^s$ will be used to denote some power of $2$. Field of order $q = 2^s$ are called (finite) binary extension fields. What are these fields and how can we construct them?

Formally, the field of order $q$, denoted $\mathbb{F}_q$, is formed by a \textit{field extension} of the binary field $\mathbb{F}_2 = \{0, 1\}$ that we discussed earlier. All this means is that $\mathbb{F}_q$ is some field that contains $\mathbb{F}_2$ inside of it. To form all of $\mathbb{F}_q$, we need to start by taking an irreducible polynomial over $\mathbb{F}_2$ of degree $s$. What is an irreducible polynomial over $\mathbb{F}_2$? A polynomial over $\mathbb{F}_2$ is a polynomial with coefficients in $\mathbb{F}_2$, and the polynomial is irreducible if it cannot be factored. For example, $x^2+1$, which is a polynomial over $\mathbb{F}_2$, is not irreducible. This is because it can be factored:
\begin{equation}
    x^2 + 1 = (x+1)(x+1).
\end{equation}
If this expression looks a little funny, remember that we are working over $\mathbb{F}_2$, so all addition takes place modulo $2$; you can check for yourself that this relation holds. Conversely, you might like to check that $x^2+x+1$ and $x^3+x^2+1$ are valid examples of irreducible polynomials over $\mathbb{F}_2$ of degrees $2$ and $3$, respectively, by proving to yourself by hand that they cannot be factored into non-trivial polynomials over $\mathbb{F}_2$.

It is a fact that for any degree $s \geq 1$, there exists an irreducible polynomial over $\mathbb{F}_2$ of degree $s$, and in general there are many such polynomials for each $s$. \href{http://fchabaud.free.fr/English/default.php?COUNT=3&FILE0=Poly&FILE1=GF(2)&FILE2=Exhaustive}{\textbf{This website}} is very helpful for listing them. It does not matter which polynomial you pick at this point; as long as the polynomial is irreducible and has degree $s$, you will end up with a field of order $q = 2^s$, and by Galois' classification of finite fields, you will always end up with the same field up to isomorphism.

Suppose we picked an irreducible polynomial of degree $s$ over $\mathbb{F}_2$ and we call it $f$. The field $\mathbb{F}_q$ is then simply formed from $\mathbb{F}_2$ by \textit{adding a root of $f$ to $\mathbb{F}_2$.} That's the entire construction, but let's see what this actually means with an example.

To form $\mathbb{F}_8 = \mathbb{F}_{2^3}$, we pick an irreducible polynomial $f$ of degree $3$ over $\mathbb{F}_2$. We have two choices, either $x^3 + x + 1$, or $x^3 + x^2 + 1$. It doesn't matter which we pick. Let's pick $f(x) = x^3 + x + 1$. Then, $\mathbb{F}_8$ is formed by simply adding in a root of $f$ to $\mathbb{F}_2$, i.e., an element $\alpha$ satisfying $f(\alpha) = 0$. What this means is that you have the elements $\mathbb{F}_2 = \{0, 1\}$, and then you add in the element $\alpha$ that satisfies this equations $f(\alpha) = \alpha^3 + \alpha + 1 = 0$, and then you ``see what happens'', i.e., you make as many elements as you can with addition and multiplication until you can't make anymore. It is a fact that you will always be able to make exactly $q$ elements ($q = 8$ in this case).

Let's see this play out. We start with the three elements $0, 1, \alpha$. By adding together $1$ and $\alpha$, we have a new element $1+\alpha$, so now we have four elements. Are we done? No, we can multiply $\alpha$ with itself to make a new element $\alpha^2$. We can then add $\alpha^2$ separately to $1, \alpha$ and $1+\alpha$ to make three more elements, $\{1+\alpha^2, \alpha + \alpha^2, 1+\alpha+\alpha^2\}$. Now, in total, we have eight elements:
\begin{equation}
    \{0, 1, \alpha, 1+\alpha, \alpha^2, 1+\alpha^2, \alpha+\alpha^2, 1+\alpha+\alpha^2\}.
\end{equation}
Can we make more? You can convince yourself quite quickly that the answer is no. You definitely can't make anything new with just addition; this is a closed group under addition. It turns out you also can't make anything new with multiplication. The reason is that $\alpha^3 + \alpha + 1 = 0$, and therefore (because addition is all mod $2$), we have $\alpha^3 = \alpha + 1$. When you try to multiply any of these two things together, you can always use the relation $\alpha^3 = \alpha + 1$ to reduce the result to one of these $8$ elements. For example, using just modulo $2$ addition, and the relation $\alpha^3 = \alpha + 1$, you can convince yourself that multiplying the element $(\alpha+\alpha^2)$ with $(1+\alpha+\alpha^2)$ gives the result $\alpha^2$.

This process allows us to construct any finite field $\mathbb{F}_q$ with $q = 2^s$ (for those interested, the process with general $q = p^s$ is entirely analogous\footnote{Note that when $p>2$, you do have to worry about the fact that subtraction is not the same thing as addition, which it is for $p = 2$.}). We simply take any irreducible polynomial $f$ over $\mathbb{F}_2$ of degree $s$, and add a root of $f$ to $\mathbb{F}_2$. The elements of $\mathbb{F}_q$ may then be simply written down as the $q = 2^s$ elements $\sum_{i=0}^{s-1}c_i\alpha^i$, where $c_i \in \mathbb{F}_2$. All the arithmetic is determined by modulo $2$ addition, and the relation $f(\alpha) = 0$. A very useful thing to convince yourself of at this point is the fact that, for any integer $A \geq 1$, and elements $\eta_1, \ldots, \eta_A \in \mathbb{F}_q$, you have
\begin{equation}\label{eq:square_distribution}
    \left(\sum_{i=1}^A\eta_i\right)^2 = \sum_{i=1}^A\eta_i^2.
\end{equation}
If you like, you can prove this by induction on $A$. In particular, we have $(\eta_1+\eta_2)^2 = \eta_1^2 + \eta_2^2$.

$\mathbb{F}_q$ is a field of order $q$, containing another field, $\mathbb{F}_2$. It is therefore correct, and indeed very helpful, to think about $\mathbb{F}_q$ as a vector space of dimension $s$ over $\mathbb{F}_2$. Indeed, you can see that the additive group of $\mathbb{F}_q$ is nothing more than the additive group of binary bit strings, which we commonly write $\mathbb{F}_2^s$; the only thing that makes it a field is now we have some rule to multiply these bit strings together (in a sensible way such that the usual rules of multiplication and division work out).

While the additive group of $\mathbb{F}_q$ is simply the group of length-$s$ bit strings under binary addition, the multiplicative group of $\mathbb{F}_q$ (which is the set of non-zero elements in $\mathbb{F}_q$, often denoted $\mathbb{F}_q^*$, under the operation of multiplication) is also quite simple. It turns out that the multiplicative group is just the cyclic group of order $q-1$. In particular, this means that there exists an element $\mu$ (there may be multiple), such that we can write the non-zero elements of $\mathbb{F}_q$ as
\begin{equation}
    \mathbb{F}_q^* = \{\mu, \mu^2, \mu^3, \mu^4, \ldots, \mu^{q-1}\}.
\end{equation}
We have the useful fact that for all $\eta \in \mathbb{F}_q$
\begin{equation}
    \eta^{q-1} = \begin{cases}
        1 &\text{ if } \eta \neq 0\\
        0 &\text{ if } \eta = 0.
    \end{cases}.\footnote{This is most easily proved using the fact that $\mathbb{F}_q^*$ is cyclic, although Lagrange's theorem for finite groups applied to $\mathbb{F}_q^*$ would give the same result regardless of the structure of $\mathbb{F}_q^*$.}
\end{equation}
Note that it follows from this fact that
\begin{equation}\label{eq:power_of_q}
    \eta^q = \eta \text{ for all } \eta \in \mathbb{F}_q.
\end{equation}

A map $f: \mathbb{F}_q \to \mathbb{F}_2$ is said to be an $\mathbb{F}_2$-linear map (often abbreviated to simply a linear map) if
\begin{align}
    f(\eta_1) + f(\eta_2) &= f(\eta_1+\eta_2) \text{ for all } \eta_1, \eta_2 \in \mathbb{F}_q\\
    f(c \cdot \eta) &= c \cdot f(\eta) \text{ for all } c \in \mathbb{F}_2 \text{ and } \eta \in \mathbb{F}_2.
\end{align}
Using linear algebra arguments, you might like to convince yourself that there are exactly $q$ such maps. One of these $\mathbb{F}_2$-linear maps is particularly famous. It is called the trace map, and written $\tr$. It is defined as follows:
\begin{equation}
    \tr(\eta) \coloneq \eta + \eta^2 + \eta^4 + \eta^8 + \ldots + \eta^{2^{s-1}}.
\end{equation}
Initially, this is a pretty odd definition. It's not obvious from this definition that $\tr$ is $\mathbb{F}_2$-linear, or that it maps into $\mathbb{F}_2$ at all. Indeed, the right-hand side is just some arithmetic in the larger field $\mathbb{F}_q$. Why is it that this particular combination of arithmetic always produces a result in $\mathbb{F}_2$? And why must it be $\mathbb{F}_2$-linear? I leave this as an exercise, where you will want to use Equation~\eqref{eq:square_distribution} liberally, and Equation~\eqref{eq:power_of_q} will be useful also. As a hint, the fact that $\tr(\eta) \in \mathbb{F}_2$ for all $\eta \in \mathbb{F}_q$ follows from the fact that $x \in \mathbb{F}_2$ if and only if $x^2 = x$. A further fact of which you might like to convince yourself, which is sometimes useful, is the fact that $\tr(\eta^2) = \tr(\eta)$ for all $\eta \in \mathbb{F}_q$.

Now that you have the trace map, it turns out that all the $\mathbb{F}_2$-linear maps $f: \mathbb{F}_q \to \mathbb{F}_2$ can be written in a very simple way. Given any such map, there exists some (unique) element $\gamma \in \mathbb{F}_q$ such that
\begin{equation}
    f(\eta) = \tr(\gamma\eta)
\end{equation}
(recall that there are exactly $q$ $\mathbb{F}_2$-linear maps $f: \mathbb{F}_q \to \mathbb{F}_2$, and so now we are able to label them by elements $\gamma \in \mathbb{F}_q$).

Once you've seen this, you will realise that actually there's nothing so special about the trace map. Indeed, you can pick any ``reference'' $\mathbb{F}_2$-linear map, call it $f^*$ and, as long as $f^*$ is not the zero map, you can write any other linear map $g$ as
\begin{equation}
    g(\eta) = f(\gamma\eta)
\end{equation}
for some $\gamma \in \mathbb{F}_q$.

\subsection{Galois Qudits}\label{subsec:galois_qudits}

The name Galois qudit is technically a bit of a misnomer. People talk about ``Galois qudits'' and ``modular qudits'' when what they're actually talking about is the same quantum system but a different choice of Pauli group.

The qudit itself is always a $q$-dimensional quantum system, so it has Hilbert space $\mathbb{C}^q$. For Galois qudits, it is canonical to label the computational basis states with elements of the finite field $\mathbb{F}_q$, so we have $q$ computational basis states
\begin{equation}
    \ket{\eta}, \text{ where } \eta \in \mathbb{F}_q.
\end{equation}
At the moment, this is just notation. Where the Galois qudit actually becomes a Galois qudit is when you make a choice of Pauli group. Let's write down the Galois qudit Pauli group for one qudit. For each $\beta \in \mathbb{F}_q$, we have a Pauli operator denoted $X^\beta$, which acts as
\begin{equation}
    X^\beta\ket{\eta} = \ket{\eta+\beta}.
\end{equation}
Furthermore, for each $\gamma \in \mathbb{F}_q$, we have a Pauli operator denoted $Z^\gamma$, which acts as
\begin{equation}
    Z^\gamma\ket{\eta} = (-1)^{\tr(\gamma\eta)}\ket{\eta}.
\end{equation}
The Pauli group for one qudit is then formed by all the operators $X^\beta$ and $Z^\gamma$, as well as their products. For this group, one may check the commutation relation
\begin{equation}
    X^\beta Z^\gamma X^\beta Z^\gamma = (-1)^{\tr(\beta\gamma)}.
\end{equation}

\begin{tcolorbox}[title=Aside on Modular Qudits]
I cannot stress enough that it is this choice of Pauli group that makes the Galois qudit the Galois qudit. By contrast, the modular qudit of dimension $q$ is a quantum system with Hilbert space $\mathbb{C}^q$, and computational basis states canonically labelled
\begin{equation}
    \ket{i}, \text{ where } i \in \{0, 1, \ldots, q-1\}.
\end{equation}
It has Pauli operators given by the ``shift and clock'' operators, also called Heisenberg-Weyl operators; for each $d \in \{0, 1, \ldots, q-1\}$,
\begin{equation}
    X^d\ket{i} = \ket{i+d \pmod q},
\end{equation}
\begin{equation}
    Z^d\ket{i} = \omega^{id}\ket{i}, \text{ where } \omega = e^{\frac{2\pi i}{q}}.
\end{equation}
This leads to a totally different Pauli group, and a totally differently behaving object. This choice of Pauli group is very well motivated from the physical point of view, but it is a little harder to treat from the standpoint of error correction (whereas the Galois qudit behaves very nicely). It is interesting to note, however, that the modular and Galois qudits coincide when $q$ is a prime.
\end{tcolorbox}

\noindent Returning to Galois qudits, we have a Hilbert space, a computational basis, and a choice of Pauli group for one qudit. What else can we define? Well, we can easily make multiple Galois qudits, by taking tensor products, and of course the Hilbert space of $n$ Galois qudits of dimension $q$ is $\left(\mathbb{C}^q\right)^{\otimes n}$. Moreover, by taking tensor products of single-qudit Pauli operators, we have the Pauli group for $n$ qudits. We denote the Pauli group for $n$ Galois qudits of dimension $q$ as $\mathcal{P}_{n,q}$.

From here, we can make the Clifford group and, more generally, a Clifford hierarchy via the usual definitions~\cite{gottesman1999demonstrating}. Indeed, as usual, the $(i+1)$-th level of the Clifford hierarchy is defined in terms of the $i$-th level, where the $1$'st level is just the Pauli group:
\begin{equation}
    \mathcal{C}_{n,q}^{(i+1)}\coloneq \{U: UPU^\dagger \in \mathcal{C}_{n,q}^{(i)} \text{ for all Paulis }P\}.
\end{equation}
In particular, $\mathcal{C}_{n,q}^{(2)}$ is the $n$-qudit Clifford group. Just as we are used to, the Pauli and Clifford groups are bona fide groups, but $\mathcal{C}_{n,q}^{(i)}$ is not a group for $i \geq 3$. However, the diagonal elements of $\mathcal{C}_{n,q}^{(i)}$ do form a group~\cite{cui2017diagonal}, as they do for qubits.

Let's make some interesting definitions of qudit gates and consider their positions in the Clifford hierarchy. First, we can consider the Galois qudit $\CNOT$ gate, which acts on $2$ qudits as
\begin{equation}
    \CNOT\ket{\eta_1}\ket{\eta_2} = \ket{\eta_1}\ket{\eta_2+\eta_1}.
\end{equation}
Note that setting $q = 2$ recovers the qubit $\CNOT$ gate. As you would expect, this is a Clifford gate. There is also a natural notion of Hadamard gate, acting as
\begin{equation}
    H\ket{\eta} = \frac{1}{\sqrt{q}}\sum_{\mu \in \mathbb{F}_q}(-1)^{\tr(\mu\eta)}\ket{\mu},
\end{equation}
which again generalises the $q=2$ case, and is always a Clifford gate. There's another natural Clifford gate for Galois qudits, which we don't really have for qubits. Considering any non-zero $\delta \in \mathbb{F}_q$, we have the ``multiplication'' gate $M^\delta$ acting as
\begin{equation}
    M^\delta\ket{\eta} = \ket{\delta\eta}.
\end{equation}
Notice that $\delta$ has to be non-zero, otherwise this gate is not even unitary. This is a bona fide Clifford gate for one Galois qudit, but we don't really think about it for qubits as the only allowed multiplication gate for qubits corresponds to $\delta = 1$, giving the identity gate.

Moving beyond the Clifford group, let us consider the qudit $\CCZ$ gate~\cite{golowich2025asymptotically,nguyen2025good}. For each $\gamma \in \mathbb{F}_q$, there is a $3$-qudit gate $\CCZ^\gamma$ acting as
\begin{equation}
    \CCZ^\gamma\ket{\eta_1}\ket{\eta_2}\ket{\eta_3} = (-1)^{\tr(\gamma\eta_1\eta_2\eta_3)}\ket{\eta_1}\ket{\eta_2}\ket{\eta_3}.
\end{equation}
One can check that for any non-zero $\gamma \in \mathbb{F}_q$, and for any $q = 2^s$, these are third level gates, and are non-Clifford (i.e., they are in exactly the third level). When $\gamma = 1$, we can just write $\CCZ^1 = \CCZ$. Notice that all $\CCZ^\gamma$ gates differ by multiplication gates (as long as $\gamma \neq 0$):
\begin{equation}
    \CCZ^\gamma = M^{\gamma^{-1}}\cdot \CCZ \cdot M^\gamma,
\end{equation}
where the multiplication gates act on any of the three qudits (say the first one). You can generalise this to $Z$ gates with many controls, for example the $Z$ gate controlled $(l-1)$-times. For example, you can consider an $l$-qudit gate $\mathsf{C}^{(l-1)}\mathsf{Z}^\gamma$, which is in the $l$'th level of the Clifford hierarchy for all $\gamma \neq 0$.

There's another interesting family of non-Clifford gates, which again we don't think about for qubits, at least directly; these were introduced in~\cite{wills2025constant}. Given $\beta \in \mathbb{F}_q$ and an integer $n \geq 1$, these are the gates $U_n^\beta$, which act on one qudit:
\begin{equation}
    U_n^\beta\ket{\eta} = (-1)^{\tr(\beta\eta^n)}\ket{\eta}.
\end{equation}
These gates can have some rather cool properties, although they're a little harder to pin down than the $\CCZ$ gates. Their level of the Clifford hierarchy unsurprisingly depends on the integer $n$, but it is more surprising that their level can also depend on the value of $q = 2^s$, and even on the choice of $\beta \in \mathbb{F}_q$. For example, you can see that for $q = 8$, $U_7^\beta$ is in exactly the third level when $\tr(\beta) = 1$, but it is just the identity gate when $\tr(\beta) = 0$.

One problem with Galois qudits, at least at the level of notation, is they don't handle phases very well, by which I mean diagonal gates with anything except $\pm 1$ on the diagonal. However, you can make the following definition for an $S$ gate. Given $\gamma \in \mathbb{F}_q$, we have
\begin{equation}
    S^\gamma\ket{\eta} = \exp\left(\frac{i\pi}{2}\tr(\gamma\eta)\right)\ket{\eta}.
\end{equation}
We can also consider a $T$ gate:
\begin{equation}
    T^\gamma\ket{\eta} = \exp\left(\frac{i\pi}{4}\tr(\gamma\eta)\right)\ket{\eta}.
\end{equation}
The reason these do not behave well is that we do not have, for example, $S^{\gamma_1}S^{\gamma_2} = S^{\gamma_1+\gamma_2}$ as you might expect.\footnote{One consequence of this is that it proves difficult to develop good codes with transversal $S$ or $T$ gates using the same methods as can give good codes with transversal $\CCZ$ or $U_n^\beta$ gates~\cite{wills2025constant,golowich2025asymptotically,nguyen2025good}.} This is not a problem if you only care about $\CSS$ codes, where we only care about ``real-valued'' gates, but worth considering in general.
\subsection{Galois Qudit CSS Codes}\label{sec:qudit_CSS_codes}
In this section, we will show that many of the tools of the stabiliser formalism~\cite{gottesman1997stabilizer} transfer naturally to the case of Galois qudits, as long as one always enforces $\mathbb{F}_q$-linearity, which we will go on to define. The enforcing of $\mathbb{F}_q$-linearity will appear as a new ingredient, although we argue that this is the natural extension of the qubit picture. Note that, for simplicity, we only treat Galois qudit $\CSS$ codes and a $\CSS$ stabiliser formalism, not more general stabiliser codes, and a more general stabiliser formalism. For the treatment of more general stabiliser codes over Galois qudits, see~\cite{ashikhmin2002nonbinary}.

A $\CSS$ code for qudits is defined by two $\mathbb{F}_q$-linear subspaces $\mathcal{L}_X \subseteq \mathbb{F}_q^n$ and $\mathcal{L}_Z \subseteq \mathbb{F}_q^n$, where $\mathcal{L}_X\subseteq \mathcal{L}_Z^\perp$. As you would expect, $\mathcal{L}_X$ will correspond to the $X$-stabilisers, and $\mathcal{L}_Z$ will correspond to the $Z$-stabilisers. However, there are two things I wish to emphasise. Firstly, when I say that $\mathcal{L} \subseteq \mathbb{F}_q^n$ is an $\mathbb{F}_q$-linear subspace, I mean that $\mathcal{L}$ satisfies the following two properties:
\begin{align}
    v_1 + v_2 &\in \mathcal{L} \text{ for all } v_1, v_2 \in \mathcal{L}\\
    c\cdot v &\in \mathcal{L} \text{ for all } c \in \mathbb{F}_q \text{ and } v \in \mathcal{L}.
\end{align}
It is the second of these properties that makes $\mathcal{L}$ an $\mathbb{F}_q$-linear subspace, rather than just, say, some additive space. Choosing our spaces of stabilisers to be $\mathbb{F}_q$-linear subspaces will make the resulting code a ``true'' stabiliser code, in Gottesman's language~\cite{gottesman2024surviving}. I will not treat ``non-true'' stabiliser codes, i.e., I will not consider the case where the stabiliser spaces are not $\mathbb{F}_q$-linear spaces, because they do not behave nicely. Secondly, I emphasise that, in the expression $\mathcal{L}_X \subseteq \mathcal{L}_Z^\perp$, the $\perp$ refers to the usual (Euclidean) inner product over $\mathbb{F}_q^n$, that is, given $v_1 \in \mathcal{L}_X$ and $v_2 \in \mathcal{L}_Z$, we have
\begin{equation}
    v_1 \cdot v_2 = \sum_{i=1}^n(v_1)_i(v_2)_i = 0,
\end{equation}
where arithmetic on the left-hand side takes place over $\mathbb{F}_q$.

Given our spaces $\mathcal{L}_X$ and $\mathcal{L}_Z$, we can make a $\CSS$ code for Galois qudits as you would expect:
\begin{equation}\label{eq:qudit_CSS_def}
    \mathsf{CSS}(\mathcal{L}_X, \mathcal{L}_Z) \coloneq \{\ket{\psi}: X^u\ket{\psi} = \ket{\psi} \text{ for all } u \in \mathcal{L}_X \text{ and } Z^v\ket{\psi} = \ket{\psi} \text{ for all } v \in \mathcal{L}_Z\}.\footnote{Note that we have made use of a common notation that, for example, given $u \in \mathbb{F}_q^n$, we define the $n$-qudit Pauli $X^u = \bigotimes_{i=1}^nX^{u_i}$.}
\end{equation}
This is a quantum code on $n$ qudits. As usual, it encodes $k$ logical qudits, where
\begin{equation}\label{eq:code_dim}
    k = n - \dim_{\mathbb{F}_q}\mathcal{L}_X - \dim_{\mathbb{F}_q}\mathcal{L}_Z.
\end{equation}
Here, I emphasise that $\dim_{\mathbb{F}_q}\mathcal{L}_X$ is the dimension of $\mathcal{L}_X$ as an $\mathbb{F}_q$-linear space, and similarly for $\dim_{\mathbb{F}_q}\mathcal{L}_Z$. Also, the distance of the code, which is the minimal number of physical qudits that an error has to touch to create a logical error, is
\begin{equation}
    d = \min(d_X, d_Z),
\end{equation}
where
\begin{align}
    d_X &= \min_{u \in \mathcal{L}_Z^\perp\setminus\mathcal{L}_X} |u|\\
    d_Z &= \min_{u \in \mathcal{L}_X^\perp\setminus\mathcal{L}_Z}|u|,
\end{align}
where $|u|$ is the $\mathbb{F}_q$-Hamming weight.

This is all very reminiscent from the qubit case. The only point I really wish to make here is that as long as you choose the case of a ``true'' stabiliser code, i.e., $\mathcal{L}_X$ and $\mathcal{L}_Z$ are $\mathbb{F}_q$-linear spaces, everything works out nicely. We will only consider such codes (all the qudit codes we regularly consider like quantum Reed-Solomon codes~\cite{aharonov1997fault,grassl1999quantum} or quantum algebraic geometry codes~\cite{ashikhmin2001asymptotically,matsumoto2001algebraic,niehage2005quantum,niehage2007nonbinary} are true $\mathbb{F}_q$-stabiliser codes).

\subsection{Qudit Pauli Measurements, Syndromes and Stabiliser Tableaux}\label{subsec:pauli_meas}

In order to operate your Galois qudit CSS code, you will need to perform some measurements to extract the syndrome of a given state under the checks of that code. To discuss this, we need to have a robust definition of what it means for a Galois qudit state to have a particular syndrome under a certain operator. 

First, to motivate the following, I just want to consider classical linear codes over $\mathbb{F}_q$ in this paragraph. Consider some classical code $C \subseteq \mathbb{F}_q^n$ defined by some parity-check matrix $H \in \mathbb{F}_q^{m \times n}$. $C$ is exactly the set of vectors $c \in \mathbb{F}_q^n$ for which $H \cdot c = 0$, where the arithmetic in this latter equation takes place over $\mathbb{F}_q$. Suppose some error happens, denoted $e \in \mathbb{F}_q^n$. We then hold $c+e \in \mathbb{F}_q^n$. To attempt to diagnose the error, we multiply this word against the parity-check matrix, to produce $H(c+e) = H\cdot e$. $H\cdot e \in \mathbb{F}_q^m$ is a syndrome, given by m entries of $\mathbb{F}_q$. Notice that each check of the code, which is a row of $H$, call it $h$, produces a syndrome component $h \cdot e \in \mathbb{F}_q$. We emphasise that the syndrome is formed of $m$ syndrome components, each being an element of $\mathbb{F}_q$.\footnote{Decoders for $\mathbb{F}_q$-linear classical codes like Reed-Solomon or algebraic geometry codes will typically work with the syndrome in this form.}

Given the discussion of the previous paragraph, the natural definition of an $n$-qudit state having a particular syndrome component under an $n$-qudit Pauli operator is as follows. Suppose that we have some pure-$X$-type or pure-$Z$-type Pauli $P$, and some state $\ket{\psi}$. We say that $\ket{\psi}$ has syndrome component $\eta$ under the $n$-qudit Pauli $P$, for some $\eta \in \mathbb{F}_q$, if
\begin{equation}
    P^\mu\ket{\psi} = (-1)^{\tr(\mu\eta)}\ket{\psi}.\footnote{For a pure-$X$-type or pure-$Z$-type $n$-qudit Pauli $P$, and $\mu \in \mathbb{F}_q$, we define $P^\mu$ as the Pauli with each of its constituent components raised to the $\mu$. For example, if $P = X_1^{\gamma_1}X_2^{\gamma_2}$, we have $P^\mu = X_1^{\mu\gamma_1}X_2^{\mu\gamma_2}$.}
\end{equation}
Note that this definition makes sense because the value $\eta \in \mathbb{F}_q$ is uniquely identifiable by all the values $\tr(\mu\eta)$ over every $\mu \in \mathbb{F}_q$, and vice versa; this is proved later in Claim~\ref{claim:value_to_trace_values}. In particular, we say that $\ket{\psi}$ has syndrome component $0$ under the operator $P$ if $P^\mu\ket{\psi} = \ket{\psi}$ for all $\mu \in \mathbb{F}_q$.  Putting multiple syndrome components together forms a syndrome; suppose we have a qudit $\CSS$ code with spaces $\mathcal{L}_X$ and $\mathcal{L}_Z$ for the $X$ checks and $Z$ checks, respectively. Noting again that these are $\mathbb{F}_q$-linear spaces, we can pick $\mathbb{F}_q$-bases of $X$ checks and $Z$ checks described respectively by collections $(v^{(X,j)})_{j=1}^{m_X}$ and $(v^{(Z,j)})_{j=1}^{m_Z}$, where $v^{(X,j)}, v^{(Z,j)} \in \mathbb{F}_q^n$. Then, the state $\ket{\psi}$ has $X$ syndrome $\sigma^{(X)} \in \mathbb{F}_q^{m_X}$ if
\begin{equation}
    X^{\mu v^{(X,j)}}\ket{\psi} = (-1)^{\tr(\mu\sigma^{(X)}_j)}\ket{\psi}, \text{ for all }\mu \in \mathbb{F}_q \text{ and } j = 1, \ldots, m_X,
\end{equation}
and similarly for the $Z$ syndrome.

Given these definitions, the usual stabiliser formalism~\cite{gottesman1997stabilizer} that we are used to goes across very nicely to the case of Galois qudits, assuming that we always stick to these definitions, in particular that we always enforce $\mathbb{F}_q$-linearity on our spaces of checks. As one example, if we consider the Hilbert space of $n$ qudits, $(\mathbb{C}^q)^{\otimes n}$, and a given Pauli $P$, demanding that your state has syndrome component $\eta \in \mathbb{F}_q$ under $P$ cuts the Hilbert space into $q$ equal pieces depending on the value of $\eta$. This can be done multiple times for independent Paulis that commute, just as we are used to for qubits (in which case $q=2$).

Sets of Galois qudit stabilisers can be tracked using stabiliser tableaux with the natural extension of the rules for qubits that we are familiar with~\cite{gottesman1997stabilizer,gottesman1998heisenberg}. Again, in this review, we only consider stabiliser tableaux of $\CSS$ type. We can start with the simplest case, that is, $\CSS$ codes of logical dimension $k=0$; these are just stabiliser states.

For qubits, we are used to specifying a stabiliser state by writing down its stabiliser tableau. When we do this, we mean that we are writing down a list of Paulis, and considering the unique state stabilised by all of them (and all of their products). For example,
\begin{equation}
    \begin{pmatrix}
        X & X & X & X\\
        Z & Z &   &  \\
          & Z & Z &  \\
          &   & Z & Z
    \end{pmatrix}
\end{equation}
is a stabiliser tableau for a $4$-qubit cat state. If no syndrome is listed, it is assumed to be trivial, but we can also list one and consider states with non-trivial syndrome. For example, the combination
\begin{equation}
    \begin{pmatrix}
        X & X & X & X\\
        Z & Z &   &  \\
          & Z & Z &  \\
          &   & Z & Z
    \end{pmatrix}\text{ with syndrome }\begin{pmatrix}
        0\\0\\1\\0
    \end{pmatrix}
\end{equation}
denotes a $4$-qubit cat state where $X$ operators have acted on the first $2$ qubits of the cat state (or equivalently the latter two), thus violating only the middle of the three $Z$-type stabilisers.

The Galois qudit generalisation of stabiliser tableaux adds an apparently new feature to make things work properly, but it is the appropriate addition to make the natural generalisation.\footnote{This ``new'' feature does not appear for qubits simply because the only non-zero element of $\mathbb{F}_2$ is $1$, the multiplicative identity.} Let's illustrate this with an example. Picking some non-zero elements $\gamma_1, \gamma_2, \gamma_3, \gamma_4 \in \mathbb{F}_q$, we can consider a stabiliser tableau
\begin{equation}\label{eq:qudit_cat_state}
    \begin{pmatrix}
        X^{\gamma_1} & X^{\gamma_2} & X^{\gamma_3} & X^{\gamma_4}\\
        Z^{\gamma_2} & Z^{\gamma_1} & &\\
        & Z^{\gamma_3} & Z^{\gamma_2} &\\
        && Z^{\gamma_4} & Z^{\gamma_3}
    \end{pmatrix}.
\end{equation}
First note that this is a valid stabiliser tableau (of $\CSS$ type) because all of the $Z$-type stabilisers and $X$-type stabilisers commute with each other, and are independent.

Let $P$ denote any of the four rows in the tableau. This stabiliser tableau defines a unique $4$-qudit state not only by demanding that the state $\ket{\psi}$ satisfies $P\ket{\psi} = \ket{\psi}$ but also by demanding that $P^\mu\ket{\psi} = \ket{\psi}$ for every $\mu \in \mathbb{F}_q$. This is the apparently new feature that we get when moving to stabiliser tableaux over $\mathbb{F}_q$; the defined state is not just stabilised by every Pauli $P$ listed, but by every Pauli $P^\mu$ for $\mu \in \mathbb{F}_q$.

Note that, this specialises to exactly what we're used to over qubits, since when $q=2$, the only choices of $\mu \in \mathbb{F}_q$ are simply $\mu = \{0,1\}$, and so the statement that $\ket{\psi}$ must be stabilised by all $P^\mu$ over $\mu \in \mathbb{F}_q$ does not change anything. Moreover, adding this apparently new feature when making the transition to higher qudits ensures everything works out well with our intuition; for example, an $n \times n$ qudit stabiliser tableau does indeed specify a unique $n$-qudit state.

Let us summarise this formally. An $n$-qudit stabiliser state (with some syndrome) may be specified by some $\mathbb{F}_q$-linear spaces $\mathcal{L}_X, \mathcal{L}_Z \subseteq \mathbb{F}_q^n$ which have $\mathbb{F}_q$-dimensions $m_X$ and $m_Z$, respectively, such that $m_X + m_Z = n$. Let us pick some $\mathbb{F}_q$-bases $\hat{\mathcal{L}}_X$ of $\mathcal{L}_X$ and $\hat{\mathcal{L}}_Z$ of $\mathcal{L}_Z$; these are sets of vectors in $\mathbb{F}_q^n$ of size $m_X$ and $m_Z$, respectively. The state may have some non-trivial $X$ syndrome $\sigma^{(X)} \in \mathbb{F}_q^{m_X}$ and $Z$ syndrome $\sigma^{(Z)} \in \mathbb{F}_q^{m_Z}$. Let us write $\hat{\mathcal{L}}_X = (v^{(j)})_{j=1}^{m_X}$ and $\hat{\mathcal{L}}_Z = (v^{(j)})_{j=1}^{m_Z}$. The stabiliser state (with syndrome) defined by $\mathcal{L}_X, \mathcal{L}_Z, \sigma^{(X)}, \sigma^{(Z)}$ is then the unique $n$-qudit state $\ket{\psi}$ satisfying
\begin{align}
    X^{\mu v^{(j)}}\ket{\psi} &= (-1)^{\tr(\mu \sigma^{(X)}_j)}\ket{\psi}\text{ for all }\mu \in \mathbb{F}_q \text{ and } j = 1, \ldots, m_X\\
    Z^{\mu v^{(j)}}\ket{\psi} &= (-1)^{\tr(\mu\sigma^{(Z)}_j)}\ket{\psi}\text{ for all }\mu \in \mathbb{F}_q \text{ and } j = 1, \ldots, m_Z.
\end{align}
The usual stabiliser tableau updates rules that allow one to track the evolution of a stabiliser state (with syndrome) through a Clifford circuit~\cite{gottesman1998heisenberg} also pass nicely to our case, as long as we always keep in mind the point about $\mathbb{F}_q$-linearity: that the full $X$-stabiliser space and $Z$-stabiliser space are only recovered by taking $\mathbb{F}_q$-linear combinations of the stabilisers listed in the tableau.

We can see an explicit example of this now. Let us see how we can use the state in Equation~\eqref{eq:qudit_cat_state} to measure the syndrome component of a stabiliser $X_1^{\gamma_1}X_2^{\gamma_2}X_3^{\gamma_3}X_4^{\gamma_4}$ on some $4$-qudit code block. Indeed, suppose that our system starts with some stabiliser tableau and syndrome
\begin{equation}
\left(\begin{array}{cccc|cccc}
X^{\gamma_1} & X^{\gamma_2} & X^{\gamma_3} & X^{\gamma_4} &   &   &   &   \\
Z^{\gamma_2} & Z^{\gamma_1} &   &   &   &   &   &   \\
  &    Z^{\gamma_3} & Z^{\gamma_2}& &   &   &   &   \\
  &   &Z^{\gamma_4} & Z^{\gamma_3} &&   &   &   \\
  &   &   &   & X^{\gamma_1} & X^{\gamma_2} & X^{\gamma_3} & X^{\gamma_4}
\end{array}\right), \begin{pmatrix}
    0\\0\\0\\0\\\eta
\end{pmatrix},
\end{equation}
where we wish to measure the value $\eta \in \mathbb{F}_q$.\footnote{Note that now we are considering an incomplete stabiliser tableau, that is, where the number of rows in the stabiliser tableau is less than $n$, the total number of qudits. In such a case, the system is thought of as having some logical information. We do not explicitly track this, but the usual rules for tracking this pass across naturally.} Here, we have written out our $4$-qudit ``cat state'' in the left-hand set of four qudits, and the $4$-qudit code block in the right-hand set of $4$ qudits, whose syndrome component $\eta$ under the operator $X_1^{\gamma_1}X_2^{\gamma_2}X_3^{\gamma_3}X_4^{\gamma_4}$ we wish to measure. Measuring that value may be achieved via a natural extension of the corresponding qubit circuit. Indeed, we are going to measure the qudit $XX$ operator on the $j$'th qudits of the cat state and the code block, for $j = 1, 2, 3, 4$. Suppose we just start by doing this for $j=1,2,3$ and we obtain the outcomes $\eta_1, \eta_2, \eta_3 \in \mathbb{F}_q$. Our stabiliser tableau has become
\begin{equation}\label{eq:stab_tab_eg_first}
\left(\begin{array}{cccc|cccc}
X^{\gamma_1} & X^{\gamma_2} & X^{\gamma_3} & X^{\gamma_4} &   &   &   &   \\
X &  &   &   & X  &   &   &   \\
  &    X & & &   &  X &   &   \\
  &   & X & &&   &  X &   \\
  &   &   &   & X^{\gamma_1} & X^{\gamma_2} & X^{\gamma_3} & X^{\gamma_4}
\end{array}\right), \begin{pmatrix}
    0\\\eta_1\\\eta_2\\\eta_3\\\eta
\end{pmatrix}.
\end{equation}
Suppose now we perform the fourth measurement, that is, we now measure the qudit $XX$ operator on the $4$th qudits of the cat state and the code block. The important point is that the outcome of this final measurement is deterministic, because this operator's eigenvalue may be determined by existing operators and their syndrome components. Indeed, the measurement of this fourth qudit $XX$ operator must have outcome $\eta_4 = \gamma_4^{-1}\eta + \gamma_4^{-1}\sum_{j=1}^3\eta_j\gamma_j$. The way to see this is that Equation~\eqref{eq:stab_tab_eg_first} defines the same state as does
\begin{equation}
\left(\begin{array}{cccc|cccc}
X^{\gamma_1} & X^{\gamma_2} & X^{\gamma_3} & X^{\gamma_4} &   &   &   &   \\
X^{\gamma_1} &  &   &   & X^{\gamma_1}  &   &   &   \\
  &    X^{\gamma_2} & & &   &  X^{\gamma_2} &   &   \\
  &   & X^{\gamma_3} & &&   &  X^{\gamma_3} &   \\
  &   &   &   & X^{\gamma_1} & X^{\gamma_2} & X^{\gamma_3} & X^{\gamma_4}
\end{array}\right), \begin{pmatrix}
    0\\\gamma_1\eta_1\\\gamma_2\eta_2\\\gamma_3\eta_3\\\eta
\end{pmatrix},
\end{equation}
which defines the same state as does
\begin{equation}
\left(\begin{array}{cccc|cccc}
 &  &  & X^{\gamma_4} &   &   &   &  X^{\gamma_4} \\
X^{\gamma_1} &  &   &   & X^{\gamma_1}  &   &   &   \\
  &    X^{\gamma_2} & & &   &  X^{\gamma_2} &   &   \\
  &   & X^{\gamma_3} & &&   &  X^{\gamma_3} &   \\
  &   &   &   & X^{\gamma_1} & X^{\gamma_2} & X^{\gamma_3} & X^{\gamma_4}
\end{array}\right), \begin{pmatrix}
    \eta+\sum_{i=1}^3\gamma_i\eta_i\\\gamma_1\eta_1\\\gamma_2\eta_2\\\gamma_3\eta_3\\\eta
\end{pmatrix},
\end{equation}
which in turn defines the same state as does
\begin{equation}
\left(\begin{array}{cccc|cccc}
 &  &  & X &   &   &   &  X \\
X^{\gamma_1} &  &   &   & X^{\gamma_1}  &   &   &   \\
  &    X^{\gamma_2} & & &   &  X^{\gamma_2} &   &   \\
  &   & X^{\gamma_3} & &&   &  X^{\gamma_3} &   \\
  &   &   &   & X^{\gamma_1} & X^{\gamma_2} & X^{\gamma_3} & X^{\gamma_4}
\end{array}\right), \begin{pmatrix}
    \gamma_4^{-1}\eta+\gamma_4^{-1}\sum_{i=1}^3\gamma_i\eta_i\\\gamma_1\eta_1\\\gamma_2\eta_2\\\gamma_3\eta_3\\\eta
\end{pmatrix}.
\end{equation}
We see that measuring the qudit $XX$ operator on the $4$'th qudits yields the claimed syndrome component $\eta_4$ deterministically. 

In total, the desired task (determining the unknown syndrome component $\eta$ on the code block) may be achieved by measuring the qudit $XX$ operators on the pairs of first, second, third, and fourth qudits, obtaining outcomes $(\eta_i)_{i=1}^4$, and calculating $\eta = \sum_{i=1}^4\gamma_i\eta_i$.

\section{Qudit-to-Qubit Mappings}\label{sec:qudit_to_qubit}

So far, we've talked about Galois qudits completely in their own right. In the lab, if you were able to manipulate a $q$-dimensional quantum system, this would be perfectly fine. However, if we are considering only physical systems based on qubits, the reason we ultimately care about Galois qudits is because of the following fact:
\begin{gather*}
    \textbf{A Galois qudit of dimension $q = 2^s$ is mathematically}\\\textbf{the same thing as a collection of $s$ qubits}.
\end{gather*}
When I say that the Galois qudit and the qubits are the ``same thing'', I mean that in every sense. It's not just that their Hilbert spaces are isomorphic: $\mathbb{C}^q \cong \left(\mathbb{C}^2\right)^{\otimes s}$, I mean that their Pauli groups are isomorphic, their Clifford groups are isomorphic, and even their Clifford hierarchies are in bijection (and their diagonal Clifford hierarchies are isomorphic, since the diagonal Clifford hierarchy forms a group~\cite{cui2017diagonal}). They really are the same thing.

More than that, the Galois qudit is an incredibly useful and natural formalism for describing collections of qubits in a neat way. Note that similar conclusions have been arrived at on the classical side, where error-correcting codes for bits are constructed by first constructing error-correcting codes over $\mathbb{F}_{2^s}$, and then operating it as an error-correcting code over bits. This is ideal for several situations that are actually faced in communication and data storage, especially when errors are bursty.

\subsection{Bases of Finite Fields}

As discussed earlier, the field $\mathbb{F}_q = \mathbb{F}_{2^s}$ may be viewed as a vector space over $\mathbb{F}_2$ of dimension $s$. This means that we may choose some basis for $\mathbb{F}_q$ over $\mathbb{F}_2$, i.e., some collection of elements $(\eta_i)_{i=0}^{s-1}$ such that any $\eta \in \mathbb{F}_q$ admits a unique decomposition
\begin{equation}
    \eta = \sum_{i=0}^{s-1}c_i\eta_i, \text{ where } c_i \in \mathbb{F}_2.
\end{equation}
Given any basis $B$ for $\mathbb{F}_q$ over $\mathbb{F}_2$, we introduce the basis decomposition map $\mathcal{D}_{B}$, which maps a field element to its components in the basis $B$, i.e.,
\begin{align}
    \mathcal{D}_{B}: \mathbb{F}_q &\to \mathbb{F}_2^s\\
    \eta = \sum_{i=0}^{s-1}c_i\eta_i &\mapsto \left(c_i\right)_{i=0}^{s-1}.
\end{align}
It is clear that for any basis $B$, $\mathcal{D}_{B}: \mathbb{F}_q \to \mathbb{F}_2^s$ is a linear isomorphism, i.e., it is a bijection, and
\begin{align}
    \mathcal{D}_{B}(\eta_1)+\mathcal{D}_{B}(\eta_2) &= \mathcal{D}_{B}(\eta_1+\eta_2)\text{ for any } \eta_1, \eta_2 \in \mathbb{F}_q\\
    \mathcal{D}_{B}(c\cdot\eta) &= c\cdot\mathcal{D}_{B}(\eta)\text{ for any }c \in \mathbb{F}_2 \text{ and } \eta \in \mathbb{F}_q.
\end{align}
Given any basis $B$, it has a \textit{dual basis}, call it $B^*$. Writing $B = (\eta_i)_{i=0}^{s-1}$ and $B^* = (\mu_i)_{i=0}^{s-1}$, these satisfy
\begin{equation}
    \tr(\eta_i\mu_j) = \delta_{ij}.
\end{equation}
The following fact is useful, and easy to show:
\begin{equation}\label{eq:trace_to_IP}
    \tr(\beta\gamma) = \mathcal{D}_{B}(\beta)\cdot \mathcal{D}_{B^*}(\gamma),
\end{equation}
where the right-hand side is the usual inner product of vectors in $\mathbb{F}_2^s$.

Now, while every basis has a dual basis, some bases have the special property that they are their own duals; they are self-dual bases. A self-dual basis always exists for $\mathbb{F}_q$ over $\mathbb{F}_2$, for any $q = 2^s$~\cite{mullen2013handbook}. For such bases $B = (\eta_i)_{i=0}^{s-1}$, we have
\begin{equation}
    \tr(\eta_i\eta_j) = \delta_{ij}.
\end{equation}
A useful property of these bases is that we can easily extract components: one can check that the $i$-th component of $\eta$ in a self-dual basis $B = \left(\eta_i\right)_{i=0}^{s-1}$ is just
\begin{equation}
    \mathcal{D}_{B}(\eta)_i = \tr(\eta\eta_i).
\end{equation}
\subsection{From Qudits to Qubits}
There are many ways to convert a qudit into a set of qubits. Each of these ways is slightly different based on which basis of $\mathbb{F}_q$ over $\mathbb{F}_2$ we pick. For this section, let us fix some basis $B = (\eta_i)_{i=0}^{s-1}$.
\subsubsection{States and Paulis}

First, we will write computational basis states of the qudit out as computational basis states of the qubits. We define a map $\varphi_{B} : \mathbb{C}^q \to \left(\mathbb{C}^2\right)^{\otimes s}$ acting as
\begin{equation}\label{eq:CBS_identify}
    \varphi_{B}\left(\ket{\eta}\right) = \ket{\mathcal{D}_{B}(\eta)},
\end{equation}
which can be extended to the whole space by linearity. Next, let's map Paulis to Paulis. We will construct a map $\Pi_{B}$ which maps Paulis on the single Galois qudit to the $s$ qubits. We identify
\begin{equation}\label{eq:X_identify}
    \Pi_{B}\left(X^\gamma\right) = X^{\mathcal{D}_{B}(\gamma)}.
\end{equation}
To clarify, $X^\gamma$ is an $X$ operator for a single Galois qudit, and $\Pi_B\left(X^\gamma\right)$ is a tensor product of $X$'s and $I$'s on $s$ qubits. We have used the usual notation that for $v \in \mathbb{F}_2^s$, $X^v \coloneq \bigotimes_{i=0}^{s-1}X^{v_i}$.

Finally, let's see how we can map $Z$ operators. We identify
\begin{equation}\label{eq:Z_identify}
    \Pi_{B}(Z^\beta) = Z^{\mathcal{D}_{B^*}(\beta)},
\end{equation}
where we stress that the right-hand side now uses the dual basis $B^*$ in the decomposition.

Let's now see that all the identifications in Equations~\eqref{eq:CBS_identify},~\eqref{eq:X_identify} and~\eqref{eq:Z_identify} behave nicely with each other. First, we can check that Equations~\eqref{eq:X_identify} and~\eqref{eq:Z_identify} define an isomorphism of the Pauli groups. Strictly speaking, defining
\begin{equation}
    \Pi_{B}(c\cdot X^\gamma Z^\beta) = c\cdot \Pi_{B}(X^\gamma)\Pi_{B}(Z^\beta)\text{ for any } c \in \mathbb{C},
\end{equation}
gives the desired full isomorphism of the Pauli groups, which may be verified by checking
\begin{equation}
    \Pi_{B}\left(X^{\gamma_1}Z^{\beta_1}X^{\gamma_2}Z^{\beta_2}\right) = \Pi_{B}\left(X^{\gamma_1}Z^{\beta_1}\right)\Pi_{B}\left(X^{\gamma_2}Z^{\beta_2}\right).
\end{equation}
We have shown that the Pauli groups of the single Galois qudit, and that of $s$ qubits are isomorphic.

We can also check that the isomorphism of states, $\varphi_{B}$, is compatible with the isomorphism of Paulis, $\Pi_{B}$, meaning that for any Pauli $P$,
\begin{equation}
    \Pi_{B}(P)\varphi_{B}\ket{\psi} = \varphi_{B}\;P\ket{\psi}.
\end{equation}
Notice that, by specifying a different basis $B$, one specifies different isomorphisms.
\subsubsection{General Operators and the Clifford Hierarchy}
One can extend $\Pi_{B}$ to an isomorphism of all unitaries acting on a single qudit, mapping into all unitaries acting on $s$ qubits. In fact, this is not special to unitaries. Any operator acting on a single Galois qudit (meaning any $q \times q$ complex matrix) can be written (uniquely) as a complex linear combination of Paulis,\footnote{One can check that the Galois qudit Paulis $X^{\gamma}Z^{\beta}$ form a basis for the complex vector space of $q \times q$ complex matrices by noting that there are $q^2$ of them, and that they are orthonormal with respect to the Hilbert-Schmidt inner product $\langle A, B \rangle = \frac{1}{q}\Tr(A^\dagger B)$. $\Tr$ is the trace of a matrix, and should not be confused with $\tr$, the trace map on finite fields.} i.e., given $U \in \mathbb{C}^{q \times q}$ (which need not necessarily be unitary), we can uniquely write
\begin{equation}
    U = \sum_{\text{Paulis }P}c_P\cdot P\text{ where } c_P \in \mathbb{C}.
\end{equation}
This allows us to extend the isomorphism $\Pi_{B}$ to all matrices:
\begin{equation}\label{eq:full_matrix_isomorphism}
    \Pi_{B}(U) = \sum_{\text{Paulis }P}c_P\cdot \Pi_{B}(P).
\end{equation}
The isomorphism of all matrices $U$ (including when $U$ is unitary) remains compatible with the isomorphism of states:
\begin{equation}\label{eq:matrix_state_compatibility}
    \Pi_{B}(U)\varphi_{B}\ket{\psi} = \varphi_{B}\;U\ket{\psi},
\end{equation}
as can be checked from the above.\footnote{This gives us the relation $\Pi_{B}(U) = \varphi_BU\varphi_B^{-1}$. We could have actually taken this as the definition and run this all differently, but I'm treating it this way because I think it emphasises the Clifford hierarchy more. However, this relation $\Pi_B(U) = \varphi_BU\varphi_B^{-1}$ is how you would want to calculate a $\Pi_B(U)$ from a given $U$ in practice.} One can also use Equation~\eqref{eq:full_matrix_isomorphism} to check that $\Pi_{B}(U_1+U_2) = \Pi_{B}(U_1)+\Pi_{B}(U_2)$ and $\Pi_{B}(U)^\dagger = \Pi_{B}(U^\dagger)$. It is slightly harder to show the following:
\begin{claim}
    $\Pi_{B}(U_1U_2) = \Pi_{B}(U_1)\Pi_{B}(U_2)$.
\end{claim}
\begin{proof}
    We want to show the operators $\Pi_{B}(U_1U_2)$ and $\Pi_{B}(U_1)\Pi_{B}(U_2)$ are equal. We do this by showing that they act in the same way on any state $\ket{\Psi} \in (\mathbb{C}^2)^{\otimes s}$. Indeed,
    \begin{align}
        \Pi_{B}(U_1U_2)\ket{\Psi}&= \Pi_{B}(U_1U_2)\varphi_{B}\;\varphi_{B}^{-1}\ket{\Psi}\\
        &=\varphi_{B}U_1U_2\varphi_{B}^{-1}\ket{\Psi}\\
        &= \varphi_{B}U_1\varphi_{B}^{-1}\varphi_{B}U_2\varphi_{B}^{-1}\ket{\Psi}\\
        &=\varphi_{B}U_1\varphi_{B}^{-1}\Pi_{B}(U_2)\ket{\Psi}\\
        &=\Pi_{B}(U_1)\Pi_{B}(U_2)\ket{\Psi}.
    \end{align}
    In the first line, we insert $\varphi_{B}\varphi_{B}^{-1}$, we use Equation~\eqref{eq:matrix_state_compatibility} going into the second line, we insert $\varphi_{B}^{-1}\varphi_{B}$ going into the third line, and we use Equation~\eqref{eq:matrix_state_compatibility} going into the fourth and fifth lines.
\end{proof}
To summarise up to this point, we have an isomorphism of states $\varphi_{B}$ which is compatible with an isomorphism of operators $\Pi_{B}$, where $\Pi_{B}$ specialises to an isomorphism on the Pauli group. For $\Pi_{B}$, we have shown the relations $\Pi_{B}(U_1+U_2) = \Pi_{B}(U_1)+\Pi_{B}(U_2)$, $\Pi_{B}(U)^\dagger = \Pi_{B}(U^\dagger)$ and $\Pi_{B}(U_1U_2) = \Pi_{B}(U_1)\Pi_{B}(U_2)$. One can manipulate these expressions to show that the same relations hold for the map $\Pi_{B}^{-1}$, which maps operators on $s$ qubits to those on a single qudit.

We can also show that $\Pi_{B}$ specialises to an isomorphism on the Clifford group. Indeed, consider a Clifford operator $C$ for a single qudit. We will show that $\Pi_{B}(C)$ is a Clifford operator for $s$ qubits. Indeed, given a Pauli $P$ on $s$ qubits, we have
\begin{align}
    \Pi_{B}(C)\cdot P\cdot\Pi_{B}(C)^\dagger &= \Pi_{B}(C)\cdot \Pi_{B}(\Pi_{B}^{-1}(P))\cdot\Pi_{B}(C^\dagger)\\
    &=\Pi_{B}(C\cdot\Pi_{B}^{-1}(P)\cdot C^\dagger).
\end{align}
In the first line, we have inserted $\Pi_{B}(\Pi_{B}^{-1}(\cdot))$ acting on $P$, and going into the second line we have used that $\Pi_{B}(U_1U_2) = \Pi_{B}(U_1)\Pi_{B}(U_2)$ for any matrices $U_1$ and $U_2$. The final result of this is some Pauli on $s$ qubits. The reason is that, because $\Pi_{B}$ specialises to an isomorphism on Paulis, $\Pi_{B}^{-1}(P)$ must be a Pauli, and so $C\cdot\Pi_{B}^{-1}(P)\cdot C^\dagger$ is a Pauli, and so $\Pi_{B}(C\cdot\Pi_{B}^{-1}(P)\cdot C^\dagger)$ is a Pauli. Therefore, $\Pi_{B}(C)$ is Clifford. One can essentially repeat this argument to show that if $C$ is a Clifford operator for a set of $s$ qubits, then $\Pi_{B}^{-1}(C)$ is a Clifford operator for a single qudit. Therefore, $\Pi_{B}$ indeed specialises to an isomorphism of Clifford groups.

In fact, extending this argument, by induction, one can show that $\Pi_{B}$ produces a bijection between the $k$-th level of the Clifford hierarchy for a single qudit and a set of $s$ qubits, for each $k$. I use the word ``bijection'' here rather than ``isomorphism'' because the $k$-th level of the Clifford hierarchy does not close as a group for $k \geq 3$. One can check, however, that $\Pi_B$ does map diagonal gates to diagonal gates, and therefore that $\Pi_B$ specialises to a genuine \textit{isomorphism} on the $k$-th level of the diagonal Clifford hierarchy, since these sets do form groups~\cite{cui2017diagonal}.

Before moving on, we can easily extend the above framework to the case of multiple qudits. Indeed, supposing we have $n$ qudits, the claim is that the $n$ qudits are exactly equivalent to a collection of $ns$ qubits (imagined as $n$ sets of $s$ qubits). Indeed, we can define a collection $\mathcal{B} = (B_i)_{i=1}^n$, where each $B_i$ is a basis for $\mathbb{F}_q$ over $\mathbb{F}_2$. We can then define an isomorphism $\varphi_{\mathcal{B}}: \left(\mathbb{C}^q\right)^{\otimes n} \to \left(\left(\mathbb{C}^2\right)^{\otimes s}\right)^{\otimes n}$, acting in the expected way. Indeed, given $u \in \mathbb{F}_q^n$, we let
\begin{equation}
    \varphi_{\mathcal{B}}\left(\ket{u}\right) = \bigotimes_{i=1}^n\varphi_{B_i}\left(\ket{u_i}\right),
\end{equation}
which extends to the whole space by linearity. We can define the isomorphism of operators in a similar way. $\Pi_{\mathcal{B}}$ acts on tensor product operators by letting $\Pi_{B_i}$ act separately on each component in the tensor product, and is extended to the whole space by linearity. We would note that an alternative, and equivalent way to define this, is to first define the map on $n$-qudit Paulis, and then extend to the whole space by linearity. That is, given $v, w \in \mathbb{F}_q^n$, we let
\begin{align}\label{eq:multi_Pauli_map_X}
    \Pi_{\mathcal{B}}(X^v) &= \bigotimes_{i=1}^nX^{\mathcal{D}_{B_i}(v_i)} = X^{\mathcal{D}_{\mathcal{B}}(v)}\\
    \Pi_{\mathcal{B}}(Z^w) &= \bigotimes_{i=1}^nZ^{\mathcal{D}_{B_i^*}(w_i)} = Z^{\mathcal{D}_{\mathcal{B}^*}(w)},\label{eq:multi_Pauli_map_Z}
\end{align}
and so on. Here, we have used the following convenient notation for the $\mathbb{F}_2$-linear isomorphism defined by the collection of bases $\mathcal{B}$:
\begin{align}
    \mathcal{D}_{\mathcal{B}}:\mathbb{F}_q^n &\to \mathbb{F}_2^{ns}\\
    v &\mapsto \bigoplus_{i=1}^n\mathcal{D}_{B_i}(v_i),
\end{align}
i.e., $\mathcal{D}_{\mathcal{B}} = \bigoplus_{i=1}^n\mathcal{D}_{B_i}$. In words, we just expand each component $v_i$ of $v \in \mathbb{F}_q^n$ in a different basis, $B_i$. Given the collection of bases $\mathcal{B} = (B_i)_{i=1}^n$, we have also used the notation $\mathcal{B}^*$ to denote the collection of bases $(B_i^*)_{i=1}^n$.

One can check that all of the previous properties carry over, like, for example, given an $n$-qudit state $\ket{\psi}$,
\begin{equation}\label{eq:n_qudit_isomorphism_compatibility}
    \Pi_{\mathcal{B}}(U)\varphi_{\mathcal{B}}\ket{\psi} = \varphi_{\mathcal{B}}\;U\ket{\psi}.
\end{equation}

\subsubsection{Stabiliser Codes}

Let us now consider how, given a $\CSS$ code for Galois qudits, we can make a $\CSS$ code for qubits. As discussed in Section~\ref{sec:qudit_CSS_codes}, we will only consider the case of ``true'' stabiliser codes, where the spaces of stabilisers form $\mathbb{F}_q$-linear spaces.

Consider, again, that we have a stabiliser code formed from $\mathbb{F}_q$-linear subspaces $\mathcal{L}_X, \mathcal{L}_Z \subseteq \mathbb{F}_q^n$ forming a quantum code $\mathsf{CSS}(\mathcal{L}_X, \mathcal{L}_Z)$ encoding $k$ logical qudits. We want to show how to turn this into a stabiliser code on $ns$ qubits encoding $ks$ logical qubits. There are actually two equivalent ways to turn this into a stabiliser code on qubits. In the first, we will map the states, and in the second, we will map the Paulis.

\paragraph{Method 1: Mapping the States} Given a collection of bases $\mathcal{B}$, and the isomorphism $\varphi_{\mathcal{B}}$, we can consider defining a quantum code just by mapping everything in the original code: $\mathcal{Q} = \varphi_{\mathcal{B}}(\mathsf{CSS}(\mathcal{L}_X, \mathcal{L}_Z))$, that is,
\begin{equation}
    \mathcal{Q} = \left\{\varphi_{\mathcal{B}}(\ket{\psi}): \ket{\psi} \in \mathsf{CSS}(\mathcal{L}_X, \mathcal{L}_Z)\right\}.
\end{equation}
Because $\varphi_{\mathcal{B}}$ is an isomorphism of complex vector spaces, it is clear that $\mathcal{Q}$ must have the same (complex) dimension as $\CSS(\mathcal{L}_X, \mathcal{L}_Z)$, that is,
\begin{align}
    \dim_{\mathbb{C}}(\mathcal{Q}) = \dim_{\mathbb{C}}\left(\CSS(\mathcal{L}_X, \mathcal{L}_Z)\right) = q^k = 2^{ks}.
\end{align}
In the second equality, we have used the fact that $\CSS(\mathcal{L}_X, \mathcal{L}_Z)$ is a code encoding $k$ logical qudits of dimension $q$, and in the third equality we have used that $q = 2^s$. Because $\dim_{\mathbb{C}}(\mathcal{Q}) = 2^{ks}$, we find that $\mathcal{Q}$ is a quantum code encoding $ks$ logical qubits, as claimed.

This method of mapping the states gives a mathematically direct definition that makes it clear that the resulting qubit code encodes $ks$ logical qubits. However, from this definition, it is not clear at all that $\mathcal{Q}$ is a stabiliser code, and if it were what its stabilisers are, or its logical operators, or anything like that. That is what Method 2 will do better.

\paragraph{Method 2: Mapping the Stabilisers}Now, given the  notation we have defined, we can consider forming subspaces $L_X, L_Z \subseteq \mathbb{F}_2^{ns}$ by mapping the subspaces $\mathcal{L}_Z$ and $\mathcal{L}_Z$:
\begin{align}
    L_X &\coloneq \mathcal{D}_{\mathcal{B}}\left(\mathcal{L}_X\right)\\
    L_Z &\coloneq \mathcal{D}_{\mathcal{B}^*}\left(\mathcal{L}_Z\right).
\end{align}
First, we claim that
\begin{align}\label{eq:dim_conversion_X}
    \dim_{\mathbb{F}_2}(\mathcal{L}_X) &= s \cdot \dim_{\mathbb{F}_q}(\mathcal{L}_X)\\
    \dim_{\mathbb{F}_2}(\mathcal{L}_Z) &= s \cdot \dim_{\mathbb{F}_q}(\mathcal{L}_Z)\label{eq:dim_conversion_Z}
\end{align}
(where $q = 2^s$), i.e., the dimensions of $\mathcal{L}_X$ and $\mathcal{L}_Z$ as $\mathbb{F}_2$-subspaces are $s$ times the dimensions of $\mathcal{L}_X$ and $\mathcal{L}_Z$ as $\mathbb{F}_q$-subspaces. Why is this?

Let us think about $\mathcal{L}_X$, which is an $\mathbb{F}_q$-linear subspace of $\mathbb{F}_q^n$. As an $\mathbb{F}_q$-linear subspace, it has dimension $K \coloneq \dim_{\mathbb{F}_q}(\mathcal{L}_X)$. Precisely, this means that any vector in $\mathcal{L}_X$ can be uniquely written as an $\mathbb{F}_q$-linear combination of a set of $K$ vectors of $\mathcal{L}_X$ (this set of $K$ vectors is called an $\mathbb{F}_q$-basis). We then ask the question: what is the size of an $\mathbb{F}_2$-basis for $\mathcal{L}_X$? The answer is $s \times K$. You can convince yourself that if $\hat{\mathcal{L}}_X \subseteq \mathcal{L}_X$ is some $\mathbb{F}_q$-basis for $\mathcal{L}_X$, so $|\hat{\mathcal{L}}_X| = K$, and $B^{(FE)}$ is some basis for $\mathbb{F}_q$ over $\mathbb{F}_2$, so $|B^{(FE)}| = s$, then $B^{(FE)} \cdot \hat{\mathcal{L}}_X$ (which is the set formed by multiplying scalars in $B^{(FE)}$ by vectors in $\hat{\mathcal{L}}_X$) is an $\mathbb{F}_2$-basis for $\mathcal{L}_X$. This set has size $s \cdot K$, and so we have demonstrated Equation~\eqref{eq:dim_conversion_X}. Equation~\eqref{eq:dim_conversion_Z} follows in the same way.

Secondly, we claim that
\begin{align}
    \dim_{\mathbb{F}_2}(L_X) &= \dim_{\mathbb{F}_2}(\mathcal{L}_X)\\
    \dim_{\mathbb{F}_2}(L_Z) &= \dim_{\mathbb{F}_2}(\mathcal{L}_Z).
\end{align}
This is simply true because of the definition of $L_X$ and $L_Z$, and the fact that $\mathcal{D}_{\mathcal{B}}$ is an isomorphism of vector spaces over $\mathbb{F}_2$.

Our next claim is that $L_X \subseteq L_Z^\perp$, where this $\perp$ symbol refers to the inner product over $\mathbb{F}_2$. To show this, consider some $\gamma \in L_X$ and $\delta \in L_Z$. By definition, $\gamma = \mathcal{D}_{\mathcal{B}}(C)$ and $\delta = \mathcal{D}_{\mathcal{B}}(D)$ for some $C \in \mathcal{L}_X$ and $D \in \mathcal{L}_Z$. We can then calculate
\begin{align}
    \gamma \cdot \delta &= \sum_{i=1}^{s\cdot n}\gamma_i\delta_i\\
    &=\sum_{i=1}^n\mathcal{D}_{B_i}(C_i)\cdot \mathcal{D}_{B_i^*}(D_i)\\
    &=\sum_{i=1}^n\tr(C_iD_i)\\
    &=\tr\left(\sum_{i=1}^nC_iD_i\right)\\
    &=\tr(0) = 0,
\end{align}
where going into the last line we have used the fact that $\mathcal{L}_X \subseteq \mathcal{L}_Z^\perp$.

Given that $L_X \subseteq L_Z^\perp$, we are at liberty to consider the $\CSS$ code for qubits, $\CSS(L_X, L_Z)$, which is a $\CSS$ code on $ns$ physical qubits. How many logical qubits does this encode? The number of logical qubits it encodes is
\begin{equation}
    ns - \dim_{\mathbb{F}_2}(L_X) - \dim_{\mathbb{F}_2}(L_Z) = s\cdot (n-\dim_{\mathbb{F}_q}(\mathcal{L}_X) - \dim_{\mathbb{F}_q}(\mathcal{L}_Z)) = s\cdot k,
\end{equation}
so that, as claimed, $\CSS(L_X, L_Z)$ encodes $sk$ logical qubits.

At this point, let us show that the code $\CSS(L_X, L_Z)$ matches the code $\mathcal{Q} = \varphi_{\mathcal{B}}(\mathsf{CSS}(\mathcal{L}_X, \mathcal{L}_Z))$ defined in method 1. The easiest way to do this is show that $\mathcal{Q} \subseteq \CSS(L_X, L_Z)$, which we do momentarily, and then to note that both spaces have dimensions (as complex vector spaces) $2^{ks}$, and so they are in fact equal.
\begin{claim}
    $\mathcal{Q} \subseteq \CSS(L_X, L_Z)$.
\end{claim}
\begin{proof}
    Let $\ket{\Psi} \in \mathcal{Q}$. By definition, $\ket{\Psi} = \varphi_{\mathcal{B}}(\ket{\psi})$ for some $\ket{\psi} \in \CSS(\mathcal{L}_X, \mathcal{L}_Z)$. Now take any $v \in L_X$ and $w \in L_Z$. By showing that $X^v\ket{\Psi} = \ket{\Psi} = Z^w\ket{\Psi}$, the proof is complete. Indeed, by definition, there exist $V \in \mathcal{L}_X$ and $W \in \mathcal{L}_Z$ such that $v = \mathcal{D}_{\mathcal{B}}(V)$ and $w = \mathcal{D}_{\mathcal{B}^*}(W)$. Then,
    \begin{align}
        X^v\ket{\Psi} &= X^v\varphi_{\mathcal{B}}\ket{\psi}\\
        &= X^{\mathcal{D}_{\mathcal{B}}(V)}\varphi_{\mathcal{B}}\ket{\psi}\\
        &=\Pi_{\mathcal{B}}(X^V)\varphi_{\mathcal{B}}\ket{\psi}\\
        &=\varphi_{\mathcal{B}}X^V\ket{\psi}\\
        &=\varphi_{\mathcal{B}}\ket{\psi}\\
        &=\ket{\Psi}.
    \end{align}
    Similarly,
    \begin{align}
        Z^w\ket{\Psi} &= Z^w\varphi_{\mathcal{B}}\ket{\psi}\\
        &=Z^{\mathcal{D}_{\mathcal{B}^*}(W)}\varphi_{\mathcal{B}}\ket{\psi}\\
        &=\Pi_{\mathcal{B}}(Z^W)\varphi_{\mathcal{B}}\ket{\psi}\\
        &=\varphi_{\mathcal{B}}Z^W\ket{\psi}\\
        &=\varphi_{\mathcal{B}}\ket{\psi}\\
        &=\ket{\Psi},
    \end{align}
    and we conclude.
\end{proof}
We have seen that we have these two equivalent methods for producing a qubit $\CSS$ code from a qudit $\CSS$ code, one which emphasises the states, and one which emphasises the stabilisers. It is also straightforward to talk about logical operators. We claim that the set of $Z$-logical operators of the qubit code can be obtained by acting with $\mathcal{D}_{\mathcal{B}^*}$ on the set of $Z$-logical operators of the qudit code, and similarly the set of $X$-logical operators of the qubit code can be obtained by acting with $\mathcal{D}_{\mathcal{B}}$ on the set of $X$-logical operators of the qubit code.\footnote{A mnemonic for remembering all the mappings is that everything that is ``$X$-type'' (stabilisers and logical operators), is mapped with $\mathcal{B}$, whereas everything that is ``$Z$-type'' is mapped with $\mathcal{B}^*$.} That is, we have the following.
\begin{claim}For the $Z$-type logical operators, we have
    \begin{align}
        L_X^\perp &= \mathcal{D}_{\mathcal{B}^*}(\mathcal{L}_X^\perp)
    \end{align}
and for the $X$-type logical operators we have
    \begin{align}
        L_Z^\perp &= \mathcal{D}_{\mathcal{B}}(\mathcal{L}_Z^\perp).
    \end{align}
\end{claim}
\begin{proof}
    We show the first statement; the second follows similarly. To show the first statement, it suffices to show first that the dimensions of the spaces are equal, and second that $\mathcal{D}_{\mathcal{B}^*}(\mathcal{L}_X^\perp) \subseteq L_X^\perp$. For the former, we have
    \begin{align}
        \dim_{\mathbb{F}_2}L_X^\perp &= ns - \dim_{\mathbb{F}_2}(L_X)\\
        &= ns - s\cdot \dim_{\mathbb{F}_q}(\mathcal{L}_X)\\
        &= s\cdot \dim_{\mathbb{F}_q}(\mathcal{L}_X^\perp)\\
        &= \dim_{\mathbb{F}_2}(\mathcal{D}_{\mathcal{B}^*}(\mathcal{L}_X^\perp)),
    \end{align}
    where the last line is true from $\dim_{\mathbb{F}_2}(\mathcal{D}_{\mathcal{B}^*}(\mathcal{L}_X^\perp)) = \dim_{\mathbb{F}_2}(\mathcal{L}_X^\perp) = s \cdot \dim_{\mathbb{F}_q}(\mathcal{L}_X^\perp)$, which follows from the exact same arguments as we made for $\mathcal{L}_X$. For the latter, we pick some $\gamma \in L_X$ and $\delta \in \mathcal{D}_{\mathcal{B}^*}(\mathcal{L}_X^\perp)$, so that there exists some $C \in \mathcal{L}_X$ and $D \in \mathcal{L}_X^\perp$ such that $\gamma = \mathcal{D}_{\mathcal{B}}(C)$ and $\delta = \mathcal{D}_{\mathcal{B}^*}(D)$. Then,
    \begin{align}
        \gamma \cdot \delta &= \sum_{i=1}^{ns}\gamma_i\delta_i\\
        &=\sum_{i=1}^n\mathcal{D}_{B_i}(C_i)\cdot\mathcal{D}_{B_i^*}(D_i)\\
        &=\sum_{i=1}^n\tr(C_iD_i)\\
        &=\tr\left(\sum_{i=1}^nC_iD_i\right)\\
        &=\tr(0)\\
        &=0,
    \end{align}
    and we conclude.
\end{proof}
\subsubsection{Measuring Qudit Paulis}

Now suppose that we have a $\CSS$ code for Galois qudits, $\CSS(\mathcal{L}_X, \mathcal{L}_Z)$, where $\mathcal{L}_X, \mathcal{L}_Z$ are $\mathbb{F}_q$-linear spaces, and $\mathcal{L}_X \subseteq \mathcal{L}_Z^\perp$. At the level of qudits, to diagnose errors, we want to measure $X^v$, where $v$ is an element of some $\mathbb{F}_q$-basis for $\mathcal{L}_X$, and $Z^w$, where $w$ is an element of some $\mathbb{F}_q$-basis for $\mathcal{L}_Z$. The outcome of each one of these qudit measurements should be an element of $\mathbb{F}_q$, which is $s$ bits of information. Suppose, for example, that we have a code state $\ket{\psi} \in \CSS(\mathcal{L}_X, \mathcal{L}_Z)$, which is afflicted by some error $X^a$, for $a \in \mathbb{F}_q^n$, so we have
\begin{equation}
    X^a\ket{\psi}.
\end{equation}
In order to diagnose the error, we wish to know all the values $w \cdot a \in \mathbb{F}_q$, where $w$ runs over an $\mathbb{F}_q$-basis for $\mathcal{L}_Z$. This is our $Z$ syndrome. How does this work when operating the code with qubits? Well, fixing some qudit stabiliser $Z^w$, if we were to just measure the qubit Pauli corresponding to $Z^w$, that is, $\Pi_{\mathcal{B}}(Z^w)$, we get one bit of information; this is not enough. Recalling Section~\ref{subsec:pauli_meas}, the syndrome component corresponding to the qudit stabiliser $Z^w$ should be some element of $\mathbb{F}_q$, which is $s$ bits of information. Fortunately, we know that $\mathcal{L}_Z$ is an $\mathbb{F}_q$-linear space, so for any scalar $\alpha \in \mathbb{F}_q$, $Z^{\alpha w}$ is also a  stabiliser. Moreover, we have
\begin{equation}
    Z^{\alpha w}X^a\ket{\psi} = (-1)^{\tr(\alpha w\cdot a)}X^a\ket{\psi}.
\end{equation}
Measuring the qubit operator corresponding
to $Z^{\alpha w}$, that is, $\Pi_{\mathcal{B}}(Z^{\alpha w})$, produces the outcome $\tr(\alpha w\cdot a)$. This means that we can learn any of the bits $\tr(\alpha w \cdot a)$ over $\alpha \in \mathbb{F}_q$. Can we learn $w \cdot a \in \mathbb{F}_q$ from this? Yes.
\begin{claim}\label{claim:value_to_trace_values}
    Given some $\rho \in \mathbb{F}_q$, knowing the $q$ values $\tr(\alpha\rho)$, where $\alpha \in \mathbb{F}_q$, is equivalent to knowing $\rho$ itself. It is also equivalent to knowing the $s$ values $\tr(b_i\rho)$ for $i = 1, \ldots, s$, where $(b_i)_{i=1}^s$ forms some basis for $\mathbb{F}_q$ over $\mathbb{F}_2$.
\end{claim}
\begin{proof}
    Clearly if we know $\rho \in \mathbb{F}_q$, then we can calculate $\tr(\alpha\rho)$ for every $\alpha \in \mathbb{F}_q$, and $\tr(b_i\rho)$. Conversely, if we know $\tr(\alpha\rho)$ for every $\alpha \in \mathbb{F}_q$, then in particular we know $\tr(b_i\rho)$ for all $i$. Now decompose $\rho$ in the dual basis: $\rho = \sum_{i=1}^sr_ib_i^*$, where $r_i \in \mathbb{F}_2$ and $B^* = (b_i^*)_{i=1}^s$. Then, a short calculation shows that $\tr(b_i\rho) = r_i$. Therefore, by knowing $\tr(b_i\rho)$, we know the components of $\rho$ in the basis $B^*$, and therefore we know $\rho$ itself.
\end{proof}
Now we see what we have to do to operate a Galois qudit code as a qubit code. Given one qudit stabiliser $Z^w$, we can learn the value $w \cdot a \in \mathbb{F}_q$ by measuring $s$ qubit Pauli operators corresponding to $Z^{b_iw}$, where $b_i$ runs over some $\mathbb{F}_2$-basis of $\mathbb{F}_q$. This gives us the desired $s$ bits of information, which may be interpreted as the $\mathbb{F}_q$-symbol  $w \cdot a$, as desired. Note that this discussion matches the discussion and definitions in Section~\ref{subsec:pauli_meas}. In short, the syndrome component of a single qudit stabiliser may be obtained by measuring $s$ qubit stabilisers.

Formally, given $\mathcal{L}_X$ and $\mathcal{L}_Z$, we pick $\mathbb{F}_q$-bases, call them $\hat{\mathcal{L}}_X$ and $\hat{\mathcal{L}}_Z$, and let us write them as $\hat{\mathcal{L}}_X = (v_j)_{j=1}^{K_X}$ and $\hat{\mathcal{L}}_Z = (w_j)_{j=1}^{K_Z}$, where $K_X \coloneq \dim_{\mathbb{F}_q}\mathcal{L}_X$ and $K_Z \coloneq \dim_{\mathbb{F}_q}\mathcal{L}_Z$. Then, for each $j = 1, \ldots, K_X$, we may choose some basis $B^{(X,j)}$ for $\mathbb{F}_q$ over $\mathbb{F}_2$, and for each $j = 1, \ldots, K_Z$, we may choose some basis $B^{(Z,j)}$ for $\mathbb{F}_q$ over $\mathbb{F}_2$. Let us write these as $B^{(X,j)} = \left(b_i^{(X,j)}\right)_{i=1}^s$ and $B^{(Z,j)} = \left(b_i^{(Z,j)}\right)_{i=1}^s$. It is helpful to denote the collections of these bases as $\mathcal{B}^X = \left(B^{(X,j)}\right)_{j=1}^{K_X}$, and $\mathcal{B}^Z = \left(B^{(Z,j)}\right)_{j=1}^{K_Z}$, respectively, which each specify the ways in which each $X$ check, and each $Z$ check, are broken up into sets of $s$ qubit checks. To operate the qudit code as a qubit code, we have $s \cdot K_X$ $X$-stabilisers to measure, which are
\begin{equation}
    \mathcal{D}_{\mathcal{B}}\left(b_i^{(X,j)}v_j\right), \text{ for } i=1, \ldots, s \text{, and } j = 1, \ldots, K_X
\end{equation}
and $s \cdot K_Z$ $Z$-stabilisers to measure, which are
\begin{equation}
    \mathcal{D}_{\mathcal{B}^*}\left(b_i^{(Z,j)}w_j\right), \text{ for } i=1, \ldots, s \text{, and } j = 1, \ldots, K_Z.
\end{equation}
Suppose for example that our code state $\ket{\psi} \in \CSS(L_X, L_Z)$ is afflicted by a $Z$-error, $Z^w$ for $w \in \mathbb{F}_2^{ns}$. We can write $w = \mathcal{D}_{\mathcal{B}^*}(W)$ for some $W \in \mathbb{F}_q^n$. When we measure the $X$-stabiliser corresponding to $\mathcal{D}_{\mathcal{B}}\left(b_i^{(X,j)}v_j\right)$, we obtain a single bit outcome
\begin{equation}
    \mathcal{D}_{\mathcal{B}}\left(b_i^{(j)}v_j\right)\cdot \mathcal{D}_{\mathcal{B}^*}(W) = \tr\left(b_i^{(j)}v_j\cdot W\right),
\end{equation}
and, according to Claim~\ref{claim:value_to_trace_values}, doing this across all $i = 1, \ldots, s$ allows us to obtain the value $v_j \cdot W \in \mathbb{F}_q$. Doing this across all $j$'s allows us to obtain a full qudit syndrome. Using a decoder for the qudit code will allow us to obtain the value of $W$ (assuming that decoding succeeds). By obtaining $W$, we can obtain the original qubit error $w = \mathcal{D}_{\mathcal{B}^*}(W)$.

\subsection{Examples}

In this section, we give some instructive examples to get us used to qudit-to-qubit mappings.

\subsubsection{Measuring a Qudit in the \texorpdfstring{$X$}{} Basis}

We begin with a simple example where we show that measuring a qudit in the $X$ basis can be achieved by measuring each of its constituent qubits in the $X$ basis, regardless of the chosen qudit-to-qubit mapping. The same statement can be made for the $Z$ basis.

This follows straightforwardly from the above prescription. Measuring a qudit in the $X$ basis means we wish to measure the qudit Pauli $X^\eta$, where $\eta = 1$. Following the above prescription, we choose any basis $B_X = \left(b^{(X)}_i\right)_{i=1}^s$ for $\mathbb{F}_q$ over $\mathbb{F}_2$, and given the qudit-to-qubit mapping for the qudit specified by some basis $B$ for $\mathbb{F}_q$ over $\mathbb{F}_2$, we wish to measure the qubit Paulis
\begin{equation}
    \mathcal{D}_B\left(b^{(X)}_i\cdot 1\right) = \mathcal{D}_B\left(b_i^{(X)}\right)\text{, for }i = 1, \ldots, s.
\end{equation}
It is quick to show that the $s$ elements $\mathcal{D}_B\left(b_i^{(X)}\right) \in \mathbb{F}_2^s$ are linearly independent bit strings over $\mathbb{F}_2$, which follows from the fact that $\left(b_i^{(X)}\right)_{i=1}^s$ are linearly independent bit strings over $\mathbb{F}_2$, and the fact that $\mathcal{D}_B$ is an $\mathbb{F}_2$-linear isomorphism. Finally, we note that measuring $s$ independent $X$-type Paulis on $s$ qubits is equivalent to measuring each of those qubits in the $X$ basis (and taking linear combinations of the measurement outcomes to obtain the desired outcomes).

\subsubsection{Measuring Qudit \texorpdfstring{$XX$}{} on Two Qudits}

As another instructive example, consider two qudits \textit{with the same qudit-to-qubit mappings}. Let us show that measuring the qudit $XX$ operator on them may be achieved by measuring the qubit $XX$ operators pairwise on their constituent qubits (which is $s$ qubit measurements).

Let the qudit-to-qubit mapping used on each qudit be described by the basis $B$ for $\mathbb{F}_q$ over $\mathbb{F}_2$, and $\mathcal{B} = (B,B)$ be two copies of this basis. Letting $B_X = \left(b_i^{(X)}\right)_{i=1}^s$ be any basis for $\mathbb{F}_q$ over $\mathbb{F}_2$, measuring the qudit $XX$ operator may be achieved by measuring the $s$ qubit operators
\begin{equation}
    \mathcal{D}_{\mathcal{B}}\left(b_i^{(X)}v\right)\text{, for }i = 1, \ldots, s,
\end{equation}
where $v = (1, 1) \in \mathbb{F}_q^2$. It may shown using similar methods to the last example that the $s$ bit strings $\mathcal{D}_{\mathcal{B}}\left(b_i^{(X)}v\right) \in \left(\mathbb{F}_2^s\right)^{\oplus 2}$ take the form $(b_i, b_i)_{i=1}^s$, where $(b_i)_{i=1}^s$ forms a basis for $\mathbb{F}_2^s$. Measuring these operators on these $2s$ qubits is equivalent to measuring the qubit operators $X_iX_i$, for $i = 1, \ldots, s$ on the two sets of qubits.

\subsubsection{Qudit \texorpdfstring{$\mathsf{CNOT}$}{} on Two Qudits}

In our final example, we will show that performing the qudit $\mathsf{CNOT}$ operation between two qudits may be achieved by performing $s$ qubit $\mathsf{CNOT}$ operations pairwise on their constituent qubits, again, as long as the qudits are decomposed with the same qudit-to-qubit mapping.

There are multiple ways one can show this. The most direct, however, is as follows. Recall from Subsection~\ref{subsec:galois_qudits} that the qudit $\mathsf{CNOT}$ operation acts on two qudits as
\begin{equation}
    \mathsf{CNOT}\ket{\eta_1}\ket{\eta_2} = \ket{\eta_1}\ket{\eta_1+\eta_2}.
\end{equation}
Let the qudit-to-qubit mapping used for both qudits be described by the basis $B = \left(b_i\right)_{i=1}^s$ for $\mathbb{F}_q$ over $\mathbb{F}_2$. Let $\mathcal{B} = (B,B)$ denote two copies of this basis. Referring to Equation~\eqref{eq:n_qudit_isomorphism_compatibility}, the qubitised form of this equation should read
\begin{equation}
    \Pi_{\mathcal{B}}(\mathsf{CNOT})\varphi_{\mathcal{B}}\left(\ket{\eta_1}\ket{\eta_2}\right) = \varphi_{\mathcal{B}}\left(\mathsf{CNOT}\ket{\eta_1}\ket{\eta_2}\right),
\end{equation}
which becomes
\begin{equation}
    \Pi_{\mathcal{B}}(\mathsf{CNOT})\varphi_{\mathcal{B}}\left(\ket{\eta_1}\ket{\eta_2}\right) = \varphi_{\mathcal{B}}\left(\ket{\eta_1}\ket{\eta_1+\eta_2}\right)
\end{equation}
after re-writing the right-hand side.
We may write $\eta_1 = \sum_{i=1}^sc^{(1)}_ib_i$ and $\eta_2 = \sum_{i=1}^sc^{(2)}_ib_i$ for some $c^{(1)}_i, c^{(2)}_i \in \mathbb{F}_2$, so that $\eta_1+\eta_2 = \sum_{i=1}^s\left(c_i^{(1)}+c_i^{(2)}\right)b_i$. We may also write $c^{(1)} = \left(c^{(1)}_i\right)_{i=1}^s$ and $c^{(2)} = \left(c^{(2)}_i\right)_{i=2}^s$, so that we have
\begin{equation}
    \Pi_{\mathcal{B}}(\mathsf{CNOT})\ket{c^{(1)}}\ket{c^{(2)}} = \ket{c^{(1)}}\ket{c^{(1)}+c^{(2)}}.
\end{equation}
This equation verifies that the qubitised form of $\mathsf{CNOT}$, which is $\Pi_{\mathcal{B}}(\mathsf{CNOT})$, acts as $s$ copies of the qubit $\mathsf{CNOT}$ gate pairwise on the constituent qubits of the two qudits.

\section{Quantum Reed-Solomon Codes}\label{sec:qRS}

There are very nice qudit codes available, with wonderful properties. We will focus on the quantum Reed-Solomon codes (qRS). The quantum Reed-Solomon codes are information-theoretically optimal codes (over Galois qudits) because they meet the quantum Singleton bound~\cite{rains2002nonbinary}, but their big drawback is that they are only available for large Galois qudits.

\subsection{Classical Reed-Solomon Codes}

We start by describing classical Reed-Solomon codes, from which quantum Reed-Solomon codes are built. 

We use the notation $\mathbb{F}_q[x]^{<k}$ to denote the set of polynomials over $\mathbb{F}_q$ of degree less than $k$ in a single variable. Choose some $n$ distinct elements $\boldsymbol{\alpha} = (\alpha_1, \ldots, \alpha_n)$, where $\alpha_i \in \mathbb{F}_q$ (note that necessarily $n \leq q$), and $n$ elements $\boldsymbol{v} = (v_1, \ldots, v_n)$, where $v_i \in \mathbb{F}_q^*$, i.e., they are all non-zero. Then, the generalised Reed-Solomon code $\mathsf{GRS}_k(\balpha, \bv)$ is defined by evaluating all the polynomials in $\mathbb{F}_q[x]^{<k}$ at all the points in $\balpha$, and component-wise multiplying by the elements of $\bv$. That is,
\begin{equation}
    \GRS_k(\balpha, \bv) \coloneq \left\{(v_1f(\alpha_1), v_2f(\alpha_2), \ldots, v_nf(\alpha_n)): f \in \mathbb{F}_q[x]^{<k}\right\}.
\end{equation}
This is an $\mathbb{F}_q$-linear code, meaning that $\GRS_k(\balpha, \bv) \subseteq \mathbb{F}_q^n$ is an $\mathbb{F}_q$-linear subspace (following from the fact that $\mathbb{F}_q[x]^{<k}$ is an $\mathbb{F}_q$-linear vector space). In particular, multiplying any codeword by any scalar in $\mathbb{F}_q$ gives you back another codeword. In addition, it has dimension $k$ and distance $d = n-k+1$~\cite{macwilliams1977theory}. This means they saturate the classical Singleton bound bound, and so they are information-theoretically optimal as classical codes.

Generalised Reed-Solomon codes behave well under taking duals. We have
\begin{equation}
    \GRS_k(\balpha, \bv)^\perp = \GRS_{n-k}(\balpha, \boldsymbol{u}),
\end{equation}
where $\boldsymbol{u} = (u_1, \ldots, u_n)$ with 
\begin{equation}\label{eq:dual_monomial}
   u_i^{-1} = v_i\prod_{j \neq i}(\alpha_i-\alpha_j). 
\end{equation}
All this says is that the generalised Reed-Solomon code with dimension $k$ evaluated at the points $\balpha$, with some multipliers $\bv$, is dual to a generalised Reed-Solomon code with dimension $n-k$ evaluated at the points $\balpha$, with some other multipliers $\boldsymbol{u}$. 

Before moving on to the quantum case, I will comment that Reed-Solomon codes are so well characterised, one may write down how many codewords they have of each Hamming weight. This is actually possible for any code meeting the classical Singleton bound (such codes are called MDS codes). We have the following (see Theorem 6 of Chapter 11 of~\cite{macwilliams1977theory}).
\begin{theorem}
    Suppose $\mathcal{C}$ is an $[n,k,d=n-k+1]$ classical code over $\mathbb{F}_q$, such as $\GRS_k(\balpha, \bv)$ as above. Then, for $w$ with $d \leq w \leq n$, the number of codewords of weight $w$ in $\mathcal{C}$ is
    \begin{equation}\label{eq:mds_weight_distbn}
        \begin{pmatrix}
            n\\w
        \end{pmatrix}(q-1)\sum_{j=0}^{w-d}(-1)^j\begin{pmatrix}
            w-1\\j
        \end{pmatrix}q^{w-d-j}.
    \end{equation}
\end{theorem}
As another example of how nice Reed-Solomon codes are, it is also possible to neatly characterise the minimum-weight codewords. For example, given the classical code $\GRS_k(\balpha,\bv)$, consider the polynomial
\begin{equation}
    (x-\alpha_1)(x-\alpha_2)\ldots(x-\alpha_{k-1}).
\end{equation}
This is a polynomial of degree less than $k$. Therefore, its evaluations at $\balpha$ (component-wise multiplied by $\bv$) will be in $\GRS_k(\balpha, \bv)$. The corresponding codeword clearly vanishes at $\alpha_1, \alpha_2, \ldots, \alpha_{k-1}$, and so is only supported on $\alpha_k , \alpha_{k+1}, \ldots, \alpha_n$ (meaning it is a codeword of weight $d = n-k+1$). We note that given any non-zero scalar $\eta \in \mathbb{F}_q^*$, we could equally-well consider the polynomial
\begin{equation}
    \eta(x-\alpha_1)(x-\alpha_2)\ldots(x-\alpha_{k-1})
\end{equation}
and the corresponding codeword would be only supported on $\alpha_k , \alpha_{k+1}, \ldots, \alpha_n$. In fact, we can actually choose any subset $\boldsymbol{\beta} \subseteq \boldsymbol{\alpha}$, where $|\boldsymbol{\beta}| = k-1$, and the polynomial
\begin{equation}
    \eta\prod_{\beta \in \boldsymbol{\beta}}(x-\beta).
\end{equation}
The evaluations of these polynomials are exactly the $(q-1)\begin{pmatrix}
    n\\n-k+1
\end{pmatrix}$ minimum-weight codewords promised to us by Equation~\eqref{eq:mds_weight_distbn}.

\subsection{Quantum Reed-Solomon Codes}

To construct a quantum Reed-Solomon code, we can start by noticing the fact that, because $\mathbb{F}_q[x]^{<k_1} \subseteq \mathbb{F}_q[x]^{<k_2}$ for any $k_1 \leq k_2$, we have
\begin{equation}
    \GRS_{k_1}(\balpha, \bv) \subseteq \GRS_{k_2}(\balpha, \bv)
\end{equation}
for any $k_1 \leq k_2$. Therefore, since we are looking for subspaces $\mathcal{L}_X$ and $\mathcal{L}_Z$ for which $\mathcal{L}_X \subseteq \mathcal{L}_Z^\perp$, we can form a valid qudit code by choosing
\begin{align}
    \mathcal{L}_X &\coloneq \GRS_{k_1}(\balpha, \bv)\\
    \mathcal{L}_Z^\perp &\coloneq \GRS_{k_2}(\balpha, \bv),
\end{align}
for some $k_1 \leq k_2$. If you like, you can denote the resulting qudit stabiliser code
\begin{equation}
    \mathsf{QRS}_{k_1, k_2}(\balpha, \bv) = \CSS(\mathcal{L}_X, \mathcal{L}_Z),
\end{equation}
recalling the definition in Equation~\eqref{eq:qudit_CSS_def}.
Given the above result on the dual of Reed-Solomon codes, we can easily write down the space of $Z$ stabilisers,
\begin{equation}
    \mathcal{L}_Z = \GRS_{n-k_2}(\balpha, \boldsymbol{u}),
\end{equation}
for some $\boldsymbol{u} = (u_1, \ldots, u_n)$ which we can easily compute using Equation~\eqref{eq:dual_monomial}. As a quantum code, the quantum Reed-Solomon code has
\begin{equation}
    k = n-k_1-(n-k_2) = k_2 - k_1
\end{equation}
logical qudits. It has $X$-distance (at the level of qudits)
\begin{equation}
    d_X = d(\mathcal{L}_Z^\perp) = n-k_2+1
\end{equation}
and $Z$-distance (at the level of qudits)
\begin{equation}
    d_Z = d(\mathcal{L}_X^\perp) = k_1+1
\end{equation}
(you can check that $d_X = d(\mathcal{L}_Z^\perp)$ and $d_Z = d(\mathcal{L}_X^\perp)$, i.e., qudit distance does not improve by considering stabilisers).

\section*{Acknowledgements}

The author is grateful for inspiring discussions with Victor V. Albert on Galois qudits. The author is grateful to Michael E. Beverland for comments on an earlier version of this review which improved its presentation. The author acknowledges funding from the MIT-IBM Watson AI Lab. This pre-print is assigned number MIT-CTP/6015.

\clearpage
\printbibliography

\end{document}